\documentclass[3p,times]{elsarticle}

\usepackage{hyperref}

\usepackage{amsmath}
\usepackage{amssymb}
\usepackage{color}
\usepackage{epstopdf}
\usepackage{times}
\usepackage[linesnumbered]{algorithm2e}
\usepackage{tikz}
\usetikzlibrary{matrix,calc,arrows,automata,plotmarks,shapes}
\usepackage{paralist}

\def\A{\mathcal{A}}
\def\B{\mathcal{B}}
\def\phi{\varphi}
\def\dist{\mathit{Dist}}
\def\dirac#1{\delta_{#1}}
\def\epsilon{\varepsilon}

\newcommand*{\dotcup}{\ensuremath{\mathaccent\cdot\cup}}

\newcommand{\TRANA}[3]{#1\xrightarrow[]{#2}#3}

\newcommand{\TRANPA}[3]{#1\xrightarrow{#2}_{{\sf P}}#3}

\newcommand{\BSP}{\ensuremath{\sim_{{\sf P}}}}

\newcommand{\AP}{\mathit{AP}}

\newcommand{\ABS}[1]{|#1|}

\renewcommand{\Pr}{\mathrm{Pr}}

\def\<{\langle}
\def\>{\rangle}
\def\l{\mathcal{L}}
\def\k{\mathcal{K}}

\newcommand{\ysim}[1]{\stackrel{#1}\sim}
\def\z{\mathbf{0}}

\def\supp{\mathit{supp}}

\def\enableAct{\mathit{EA}}
\newcommand{\consistDist}[1]{\overrightarrow{#1}}
\def\jan{\dag}
\def\late{\S}
\def\lateBisimulation{\;\sim_{\late}\;}
\def\janBisimulation{\;\sim_{\jan}\;}
\def\nlateBisimulation{\;\not\sim_{\late}\;}
\def\njanBisimulation{\;\not\sim_{\jan}\;}
\newcommand{\janTran}[3]{{#1}~{{\stackrel{#2}\rightsquigarrow}}~{#3}}

\def\actSet{\mathbb{A}}

\def\traceEquiv{\simeq}
\def\traceEquivPriori{\simeq_{\mathit{prio}}}
\def\trace{w}
\def\path{\sigma}
\def\scheduler{\pi}
\def\pathSet{\mathit{Paths}}
\newcommand{\lastState}[1]{#1\downarrow}
\def\concat{\circ}

\DeclareMathAlphabet{\mathpzc}{OT1}{pzc}{m}{it}
\newcommand{\distributedSchedulers}{{{\mathpzc S}_D}}
\newcommand{\APAR}[3][\actSet]{{#2}\parallel_{#1}{#3}}

\newcommand{\distTran}[3]{#1\stackrel{#2}{\hookrightarrow}#3}

\newtheorem{theorem}{Theorem}[section]
\newtheorem{corollary}{Corollary}[section]
\newtheorem{lemma}{Lemma}[section]

\newtheorem{definition}{Definition}[section]
\newtheorem{proposition}{Proposition}[section]
\newtheorem{example}{Example}[section]
\newenvironment{proof}{{\sc Proof. }}{}

\begin{document}
\begin{frontmatter}

\title{
Distribution-based Bisimulation and Bisimulation Metric\\ in Probabilistic
  Automata\tnoteref{titlenote} 
}
\tnotetext[titlenote]{Supported by the National Natural Science Foundation of
    China (NSFC) (Grant No.  61428208, 61472473 and 61361136002), 
    Australian Research Council (ARC) under grant No
    DP130102764, and by the Overseas
    Team Program of Academy of Mathematics and Systems Science, CAS
    and the CAS/SAFEA International Partnership Program for
    Creative Research Team.}

\author[yuan]{Yuan Feng}
\address[yuan]{University of Technology Sydney, Australia}

\author[lei,leis]{Lei Song}
\address[lei]{Saarland University, Germany}
\address[leis]{Max-Planck-Institut f\"{u}r Informatik, Germany}

\author[lijun]{Lijun Zhang}
\address[lijun]{State Key Laboratory of Computer Science, Institute of Software, Chinese Academy of Sciences, China}

\begin{abstract}
  Probabilistic automata were introduced by Rabin in 1963 as language
  acceptors. Two automata are equivalent if and only if they accept
  each word with the same probability.  On the other side, in the
  process algebra community, probabilistic automata were re-proposed
  by Segala in 1995 which are more general than Rabin's
  automata. Bisimulations have been proposed for Segala's automata to
  characterize the equivalence between them. So far the two notions of
  equivalences and their characteristics have been studied mostly
  independently.  In this paper, we consider Segala's automata, and
  propose a novel notion of distribution-based bisimulation by joining the existing
  equivalence and bisimilarities.  We demonstrate the utility of our
  definition by studying distribution-based bisimulation metrics,
  which gives rise to a robust notion of equivalence for Rabin's
  automata. We compare our notions of bisimulation to some existing
  distribution-based bisimulations and discuss their compositionality
  and relations to trace equivalence. Finally, we show the decidability
  and complexity of all relations. 
\end{abstract}


\end{frontmatter}


\section{Introduction}
In 1963, Rabin \cite{Rabin63} introduced the model \emph{probabilistic
  automata} as language acceptors. In a probabilistic automaton, each
input symbol determines a stochastic transition matrix over the state
space. Starting with the initial distribution, each word (a sequence
of symbols) has a corresponding probability of reaching one of the
final states, which is referred to as the accepting probability. Two
automata are equivalent if and only if they accept each word with the
same probability. The corresponding decision algorithm has been
extensively studied, see~\cite{Rabin63,Tzeng92,KieferMOWW11,KieferMOWW12}.

Markov decision processes (MDPs) were known as early as the
1950s~\cite{Bellman57}, and are a popular modelling formalism used for
instance in operations research, automated planning, and decision
support systems. In MDPs, each state has a set of enabled actions and
each enabled action leads to a distribution over successor
states. MDPs have been widely used in the formal verification of
randomized concurrent systems \cite{BiancoA95}, and are now supported by probabilistic
model checking tools such as PRISM \cite{KwiatkowskaNP11},
MRMC~\cite{KatoenZHHJ11} and IscasMC~\cite{HLSTZ14}.

On the other side, in the context of concurrent systems, probabilistic
automata were re-proposed by Segala in 1995~\cite{Segala-thesis},
which extend MDPs with internal nondeterministic choices.  Segala's
automata are more general than Rabin's automata, in the sense that
each input symbol may correspond to more than one stochastic
transition matrices. Various behavioral equivalences were defined,
including strong bisimulations, strong probabilistic bisimulations,
and weak bisimulation extensions~\cite{Segala-thesis}. Strong
bisimulations require all transitions being matched by equivalent
states, whereas weak bisimulations allow single
transition being matched by a finite execution fragment. These
behavioral equivalences are used as powerful tools for state space
reduction and hierarchical verification of complex systems. Thus,
their decision algorithms~\cite{CattaniS02,BEM00,HermannsT12} and
logical
characterizations~\cite{ParmaS07,DesharnaisGJP10,HermannsPSWZ11} were
widely studied in the literature.

Equivalences are defined for the specific initial distributions over
Rabin's automata which are \emph{deterministic} with respect to the
input word. On the other side, bisimulations are usually defined over
states over Segala's automata which are \emph{non-deterministic},
i.e., an input word induces often more than one probability
distributions. 
For Segala's automata, state-based bisimulations have arguably too
strong distinguishing power, thus various
relaxations have been proposed, recently. 
In~\cite{DoyenHR08}, a distribution-based bisimulation
was defined for Rabin's automata. Because of the deterministic nature,
they proved further that this turns out to be an equivalent characterization of
the equivalence in a coinductive manner, as for bisimulations.  
Later,
this leads to, a
distribution-based weak bisimulation in~\cite{EisentrautHZ10}.  Since
states can be matched by distributions, this novel
equivalence relation ignores the branching structure of the model, and produces
weaker relations over states than classical weak bisimulations
\cite{Segala-thesis}. 

The induced distribution-based strong bisimulation was further
studied in~\cite{Hennessy12}. Interestingly, whereas the weak
bisimulation is weaker, it was shown that the distribution-based strong
bisimulation agrees with the state-based bisimulations when lifted to
distributions. This is rooted in the formulation of the weak
bisimulation in \cite{EisentrautHZ10}: constraints are proposed for
each state in the distribution independently. This may appear too fine
for the overall behaviour of the distribution: this limitation was
illustrated in an example at the end of the paper
\cite{EisentrautHZ10}.
Other formulations of distribution-based bisimulations have been
proposed recently.  In \cite{EisentrautGHSZ12}, weaker notions of weak
bisimulations are further studied. But the corresponding strong
bisimulations have not been discussed. Recently, in \cite{HermannsKK14},
a new distribution-based strong bisimulation is proposed for Segala's
automata and extensions with continuous state space.

As one contribution of this paper, we consider Segala's probabilistic
automata, and propose a novel notion of strong distribution-based
bisimulation. The novel relation is coarser than the relations in
\cite{Hennessy12,EisentrautHZ10}: we show that for Rabin's
probabilistic automata it coincides with equivalences, and for Segala's probabilistic automata,
it is reasonably weaker than the existing bisimulation relations.
Thus, it joins the two equivalence notions restricting to the corresponding
sub-models.

Another contribution of this paper is the characterization of
distribution-based bisimulation metrics. Bisimulations for
probabilistic systems are known to be very sensitive to the transition
probabilities: even a tiny perturbation of the transition
probabilities will destroy bisimilarity. Thus, bisimulation metrics
have been proposed~\cite{GiacaloneJS90}: the distance between any two
states are measured, and the smaller the distance is, the more similar
they are. If the distance is zero, one then has the classical
bisimulation. Because of the nice property of robustness,
bisimulation metrics have attracted a lot attentions on MDPs and their
extensions with continuous state space, see
\cite{DesharnaisGJP99,AlfaroMRS07,DesharnaisGJP04,DesharnaisLT08,FernsPP11,ChenBW12,Fu12,BacciBLM13,ComaniciPP12,ChatzikokolakisGPX14}.
All of the existing bisimulation metrics mentioned above are
state-based. On the other side, as states lead to distributions in
MDPs, the metrics must be lifted to distributions. In the second part
of the paper, we propose a distribution-based bisimulation metric; we
consider it being more natural as no lifting of distances is
needed. We provide a coinductive definition as well as a fixed point
characterization, both of which are used in defining the state-based
bisimulation metrics in the literature. We provide a logical
characterization for this metric as well, and discuss the relation of
our definition and the state-based ones.
A direct byproduct of our bisimulation-based metrics is the notion of
equivalence metric for Rabin's probabilistic automata. As for
bisimulation metrics, the equivalence metric provides a robust
solution for comparing Rabin's automata. To the best of our knowledge,
this has not been studied in the literature.

Lastly, we consider the strong bisimulation obtained from
\cite{EisentrautGHSZ12}, and investigate the relations of it, together with our relation and the one defined in  \cite{HermannsKK14}. We show that our distribution-based bisimulation is the
coarsest among them. In more detail, we show that even though all the
distribution-based bisimulations are not compositional in general,
they are compositional if restricted to a subclass of schedulers.  We
compare all notions of bisimulation to trace equivalences, and show
that they are all finer than a priori trace distribution equivalence,
but incomparable to trace equivalences.
Further, we study the corresponding decision algorithms in this paper.
In \cite{HermannsKK14}, the authors have studied the decision
algorithm for their bisimulation, and they pointed out that it can be used
for deciding other variants as well.  In contrast to state-based
bisimulation metrics, the problem of computing distribution-based
bisimulation metrics is much harder. Actually, we prove that the
problem of computing the distribution-based bisimulation metrics
without discounting is undecidable. However, if equipped with a
discounting factor, the problem turns to be decidable and is NP-hard.

This paper is an extended version of the conference
paper~\cite{FengZ13}. In addition to the conference version, we have
added missing proofs, 
investigated the relationship of our bisimulation with other
distribution-based bisimulations in the literature, and shown
their compositionality with respect to a subclass of
schedulers. We also discussed
the decidability and complexity of computing all mentioned relations.


\emph{Organization of the paper.}  We discuss related works in Section \ref{sec:rel}. We introduce some notations in
Section \ref{sec:pre}. Section \ref{sec:pa} recalls the definitions of
probabilistic automata, equivalence, and bisimulation relations. We
present our distribution-based bisimulation in Section
\ref{sec:novel}, and bisimulation metrics and their logical
characterizations in Section~\ref{sec:metric}. In Section~\ref{sec:relation-literature} we recall some
existing notions of bisimulation in the literature and compare them
with our bisimulation. We also discuss their compositionality and
relations to trace equivalences. The decidability and complexity of  distribution-based 
approximate bisimulation are presented in
Section~\ref{sec:decision-algorithms}, while Section \ref{sec:conclusion}
concludes the paper. 

\section{Related Works}\label{sec:rel}
In probabilistic verification, bisimulation-based behavioral
equivalences are often used in abstracting the original system by
aggregating bisimilar states together.  Coarser bisimulation thus
leads to smaller quotient system through the aggregation process. On
the other side, smaller quotient may lose properties of the
original system. The logical characterization problem studies the
relationship between \emph{bisimilar states} and \emph{logical equivalent states}.

For Segala's automata, he has already investigated the relationship of
behavioural equivalences and logical equivalences with respect to the
logic PCTL (probabilistic computational tree logic).  It was shown that
strong bisimulation preserves PCTL properties, and weak
bisimulation preserves a PCTL fragment without the next operator
\cite{Segala-thesis,SegalaL95}. Moreover, these bisimulations are strictly finer than PCTL equivalence,
i.e., they distinguish even states which satisfy the same set of
PCTL formulas. In \cite{SongZG11}, a novel coarser
bisimulation was proposed which agrees with the PCTL logical
equivalences.  Extensions of the Hennessy-Milner logic of Larsen and
Skou~\cite{LarsenS91} were also extensively studied in the literature in this
respect, including~\cite{Jonsson,ParmaS07,DesharnaisGJP10,DArgenioWTC09,BernardoNL13}.


PCTL logical formulas have atomic propositions to characterize state
properties, and can express more involved nested properties. For a
simple class of properties, such as the probabilistic reachability,
the \emph{state-based} bisimulations are arguably too fine grained.
In the literature several authors proposed preorders based on
asymmetric simulation relations \cite{SegalaL95,Sto99}. Recently, this has led to further development
of several distribution-based symmetric bisimulations
\cite{EisentrautHZ10,Hennessy12,EisentrautGHSZ12,FengZ13,HermannsKK14},
as discussed in the introduction.

To construct the quotient system with respect to a bisimulation, one
needs to decide whether two states or distributions are
bisimilar. Thus, decision algorithm for bisimulations
is a fundamental problem, and has been extensively studied in the
literature.  This rooted in the \emph{partition refinement algorithm}
for the classical transition system.  For Segala's automata, while
state-based bisimulation can be decided in polynomial time
\cite{BEM00,CattaniS02,ZHEJ07,CrafaR12,HermannsT12}, decision
procedures for distribution-based bisimulation are more expensive than
the ones for state-based bisimulation \cite{SchusterS12,EisentrautHKT013}. 

\section{Preliminaries}\label{sec:pre}
\paragraph{Distributions.}
For a finite set $S$, a (probability) distribution is a function $\mu:S\to
[0,1]$ satisfying $\ABS{\mu}:=\sum_{s\in S}\mu(s)= 1$. We denote by
$\mathit{Dist}(S)$ the set of distributions over $S$. We shall use
$s,r,t,\ldots$ and $\mu,\nu\ldots$ to range over $S$ and
$\mathit{Dist}(S)$, respectively. Given a set of distributions
$\{\mu_i\}_{1\leq i\leq n}$, and a set of positive weights
$\{p_i\}_{1\leq i\leq n}$ such that $\sum_{1\leq i\leq n}p_i=1$, the
\emph{convex combination} $\mu=\sum_{1\leq i\leq n}p_i\cdot\mu_i$ is
the distribution such that $\mu(s)=\sum_{1\leq i\leq
  n}p_i\cdot\mu_i(s)$ for each $s\in S$. The support of $\mu$ is
defined by $\mathit{supp}(\mu):=\{s\in S \mid \mu(s)>0\}$. For an
equivalence relation $R$ defined on $S$, we write $\mu R\nu$ if it holds
that $\mu(C)=\nu(C)$ for all equivalence classes $C\in S/R$. A
distribution $\mu$ is called \emph{Dirac} if $|\mathit{supp}(\mu)|=1$,
and we let $\dirac{s}$ denote the Dirac distribution with
$\dirac{s}(s)=1$.

Note that when $S$ is finite and an order over $S$ is fixed, the distributions
$\dist(S)$ over $S$, when regarded as a subset of $\mathbb{R}^{|S|}$,
is both convex and compact. In this paper, when we talk about convergence of distributions,
or continuity of relations such as transitions, bisimulations, and pseudo-metrics 
between distributions, we are referring to the normal topology of $\mathbb{R}^{|S|}$.
For a set $F\subseteq S$, we define the characteristic
(column) vector $\eta_F$ by letting $\eta_F(s)=1$ if $s\in
F$, and 0 otherwise.

\paragraph{pseudo-metric.}
A pseudo-metric over $\dist(S)$ is a function $d : \dist(S)\times \dist(S) \rightarrow [0,1]$ such that 
(i) $d(\mu, \mu) = 0$;
(ii) $d(\mu, \nu) = d(\nu, \mu)$;
(iii) $d(\mu, \nu) + d(\nu, \omega) \geq d(\mu, \omega)$.
In this paper, we assume that a pseudo-metric is continuous.

\section{Probabilistic Automata and Bisimulations}\label{sec:pa}
\subsection{Probabilistic Automata}
Let $AP$ be a finite set of atomic propositions. We recall the notion of 
probabilistic automata introduced by Segala~\cite{Segala-thesis}.
\begin{definition}[Probabilistic Automata]\label{def:automata}
  A \emph{probabilistic automaton} is a tuple
  $\A=(S, Act, \rightarrow,L,\alpha)$ where $S$ is a finite set of
  states, $Act$ is a finite set of actions, ${\rightarrow}\subseteq S\times Act\times \dist(S)$ is a
  transition relation, $L:S\to 2^{AP}$ is a labelling function, and $\alpha\in \dist(S)$ is the initial distribution.
\end{definition}

As usual we only consider image-finite probabilistic automata, i.e.
 for all $s\in S$, the set $\{\mu\mid (s,a,\mu)\in{\rightarrow}\}$ is finite. A
transition $(s,a,\mu)\in{\rightarrow}$ is denoted by
$\TRANA{s}{a}{\mu}$.  We denote by $\enableAct(s):=\{a\mid
\TRANA{s}{a}{\mu}\}$ the set of enabled actions in $s$. We say $\A$ is
\emph{input-enabled}, if $\enableAct(s)=Act$ for all $s\in S$. The state $s$ is \emph{deterministic} if  $(s,a,\mu)\in{\rightarrow}$ and  $(s,a,\mu')\in{\rightarrow}$  imply that $\mu=\mu'$. 
We say $\A$ is
an MDP if all states in $S$ are deterministic.

Interestingly, a subclass of probabilistic automata were already
introduced by Rabin in 1963~\cite{Rabin63}; Rabin's probabilistic automata
were referred to as \emph{reactive automata} in~\cite{Segala-thesis}. We adopt this
convention in this paper.
\begin{definition}[Reactive Automata]
  We say $\A$ is \emph{reactive} if it is input-enabled and deterministic, and for all $s$, $L(s)\in
  \{\emptyset,AP\}$.
\end{definition}

Here the condition $L(s)\in \{\emptyset,AP\}$ implies that the states
can be partitioned into two equivalence classes according to their
labelling. Below we shall identify $F:=\{s\mid L(s)=AP\}$ as the set
of \emph{accepting states}, a terminology used in automata theory.
In a reactive automaton, each action $a\in Act$ is enabled
precisely once for all $s\in S$, thus inducing a
stochastic matrix $M(a)$ satisfying $\TRANA{s}{a}{M(a)(s,\cdot)}$.

\subsection{Probabilistic Bisimulation and Equivalence}
First, we recall the definition of (strong) probabilistic
bisimulation for probabilistic automata~\cite{Segala-thesis}.
Let $\{\TRANA{s}{a}{\mu_i}\}_{i\in I}$ be a collection of transitions, and let $\{p_i\}_{i\in I}$ be a collection of probabilities
with $\sum_{i\in I}p_i=1$. Then $(s,a,\sum_{i\in I}p_i\cdot\mu_i)$ is
called a \emph{combined transition} and is denoted by
$\TRANPA{s}{a}{\mu}$ where $\mu=\sum_{i\in I}p_i\cdot\mu_i$.  

\begin{definition}[Probabilistic bisimulation~\cite{Segala-thesis}]\label{def:bis}
  An equivalence relation $R\subseteq S\times S$ is a 
  probabilistic bisimulation if $sR r$ implies that $L(s)=L(r)$, and for each
  $\TRANA{s}{a}{\mu}$, there exists a combined transition
  $\TRANPA{r}{a}{\nu}$ such that $\mu R \nu$.

We write $s\BSP r$ whenever there is a 
probabilistic bisimulation $R$ such that $sRr$.
\end{definition}

Recently, in~\cite{EisentrautHZ10}, a distribution-based weak
bisimulation has been proposed, and the induced distribution-based
strong bisimulation is further studied in~\cite{Hennessy12}. Their
bisimilarity is shown to be the same as $\BSP$ when lifted to
distributions. Below we recall the definition of equivalence for
reactive automata introduced by Rabin~\cite{Rabin63}.

\begin{definition}[Equivalence for Reactive Automata~\cite{Rabin63}]\label{def:equiv-react-auto}
  Let $\A_i=(S_i,Act,\rightarrow_i,L_i,\alpha_i)$, $i=1,2$, be
  two reactive automata with
  $F_i$ being the set of final states for $\A_i$. We say
  $\A_1$ and $\A_2$ are equivalent if $\A_1(w)=\A_2(w)$
  for each $w\in Act^*$, where $\A_i(w):=\alpha_i M_i(a_1)\ldots M_i(a_k)
  \eta_{F_i}$ provided $w=a_1\ldots a_k$.
\end{definition}

Stated in plain english, $\A_1$ and $\A_2$ with the same set of
actions are equivalent iff for an arbitrary input $w$, 
$\A_1$ and $\A_2$ accept $w$ with the same probability. 

So far bisimulations and equivalences were studied mostly
independently.  The only exception we are aware is~\cite{DoyenHR08}, in which for
Rabin's probabilistic automata, a distribution-based bisimulation is
defined that generalizes both equivalence and bisimulations. 
\begin{definition}[Bisimulation for Reactive
  Automata~\cite{DoyenHR08}] \label{def:distribution_based} Let
  $\A_i=(S_i,Act,\rightarrow_i,L_i,\alpha_i)$, $i=1,2$, be two
  reactive automata, and $F_i$ the set
  of final states for $\A_i$. A relation $R\subseteq
  \mathit{Dist}(S_1) \times \mathit{Dist}(S_2)$ is a bisimulation if
  for each $\mu R\nu$ it holds (i) $\mu\cdot\eta_{F_1} = \nu \cdot
  \eta_{F_2}$, and (ii) $(\mu M_1(a)) R (\nu M_2(a))$ for all $a\in
  Act$.

We write $\mu\sim_d\nu$ whenever there is a 
bisimulation $R$ such that $\mu R\nu$.
\end{definition}

It was shown in~\cite{DoyenHR08} that two reactive automata are equivalent if
and only if their initial distributions are distribution-based
bisimilar according to the definition above.

\section{A Novel Bisimulation Relation}\label{sec:novel}
In this section we introduce a notion of distribution-based bisimulation
for Segala's automata by extending the bisimulation defined
in~\cite{DoyenHR08}. We shall show
the compatibility of our definition with previous ones in Subsection~\ref{sec:compatibility}, and
some properties of our bisimulation in Subsection~\ref{sec:property}.

For the first step of defining a distribution-based bisimulation, we need to extend the transitions 
starting from states to those starting from distributions. A natural candidate for such an extension is 
as follows: for a distribution $\mu$ to perform an action $a$, each
state in its support \emph{must} make a combined $a$-move. However,
this definition is problematic, as in Segala's general probabilistic automata,
action $a$ may not always be enabled in all support states of $\mu$.
In this paper, we deal with this problem by 
first defining distribution-based bisimulations (resp. distances) for input-enabled automata, for which
the transition between distributions can be naturally defined, and then
reducing the equivalence (resp. distance) of two distributions
in a general probabilistic automaton to the bisimilarity (resp. distance) of these distributions
in an input-enabled automaton which is obtained from the original one
by adding a \emph{dead} state and some additional transitions to the
dead state.

To make our idea more rigorous, we need some notations. For
$A\subseteq AP$ and a distribution $\mu$, we define
$\mu(A):=\sum \{ \mu(s) \mid L(s) =A \}$, which is the probability of being in those state $s$ with label $A$.

 \begin{definition}\label{def:dptran}
   We write $\TRANA{\mu}{a}{\mu'}$ if for each
   $s\in \supp(\mu)$ there exists
   $\TRANPA{s}{a}{\mu_s}$ such that 
   $\mu'=\sum_s\mu(s)\cdot\mu_s$.
\end{definition}

We first present our distribution-based bisimulation for input-enabled probabilistic automata.

\begin{definition}\label{def:bisi}
  Let $\A=(S, Act, \rightarrow,L,\alpha)$ be an input-enabled probabilistic automaton.  A symmetric
  relation $R \subseteq \dist(S)\times \dist(S)$ is a (distribution-based) bisimulation if $\mu R\nu$ implies that
  \begin{enumerate}
  \item $\mu(A)=\nu(A)$ for each $A\subseteq AP$, and
  \item for each $a\in Act$, whenever $\TRANA{\mu}{a}{\mu'}$ then there exists a transition
  $\TRANA{\nu}{a}{\nu'}$ such that $\mu' R\nu'$.
  \end{enumerate}
We write $\mu\sim^\A\nu$ if there is a 
  bisimulation $R$ such that $\mu R \nu$.
\end{definition}

Obviously, the bisimilarity $\sim^\A$ is the largest bisimulation relation over $\dist(S)$. 

For probabilistic automata which are not input-enabled, we define distribution-based bisimulation with the help of
\emph{input-enabled extension} specified as follows.

\begin{definition}\label{def:inp-enb-ext}
Let $\A=(S, Act, \rightarrow,L,\alpha)$ be a probabilistic automaton over $AP$.
The \emph{input-enabled extension} of $\A$, denoted by $\A_\bot$, is defined as
an (input-enabled) probabilistic automaton $(S_\bot, Act, \rightarrow_\bot, L_\bot, \alpha)$ over $AP_\bot$ where
\begin{enumerate}
\item  $S_\bot = S\cup \{\bot\}$ where $\bot$ is a \emph{dead} state
not in $S$;
\item $AP_\bot=AP\cup \{dead\}$ with $dead\not\in AP$;
\item $\rightarrow_\bot {=} \rightarrow \cup~ \{(s, a, \dirac{\bot}) \mid a\not \in \enableAct(s)\} \cup \{(\bot, a, \dirac{\bot}) \mid a \in Act\}$;
\item $L_\bot(s) = L(s)$ for any $s\in S$, and $L_\bot(\bot)=\{dead\}$.
\end{enumerate} 
\end{definition}

\begin{definition}\label{def:bisimulation}
  Let $\A$ be a probabilistic automaton which is not input-enabled. Then 
  $\mu$ and $\nu$ are bisimilar, denoted by $\mu\sim^\A\nu$, 
  if $\mu\sim^{\A_\bot}\nu$.
\end{definition}

We always omit the superscript $\A$ in $\sim^\A$ when no confusion arises.

\subsection{Compatibility}\label{sec:compatibility}
In this subsection we instantiate appropriate labelling functions and show
that our notion of bisimilarity is a conservative extension of both
probabilistic bisimulation~\cite{Rabin63} 
and equivalence relations~\cite{DoyenHR08}.

\begin{lemma}\label{lem:bsp}
 Let $\A$ be a probabilistic automaton where $AP=Act$, and $L(s) =\enableAct(s)$ for each $s$. Then, $\mu\BSP \nu$ implies
 $\mu \sim \nu$.
\end{lemma}
\begin{proof}
First, it is easy to see that for a given probabilistic automaton $\A$ with $AP=Act$ and $L(s) =\enableAct(s)$ for each $s$, and distributions $\mu$ and $\nu$ in $\dist(S)$, 
$\mu \BSP\nu$ in $\A$ if and only if $\mu \BSP\nu$ in the input-enabled extension $\A_\bot$. Thus 
we can assume without loss of any generality that $\A$ itself is input-enabled. 

It suffices to show that the symmetric relation
$$R = \{ (\mu, \nu) \mid \mu \BSP\nu \}$$
is a bisimulation. For each $A\subseteq Act$, let $S(A) = \{s\in S \mid  L(s)= A\}$. Then $S(A)$ is the disjoint union of some equivalence classes of $\BSP$; that is, $S(A) = \dotcup\{M\in S/{\BSP} \mid M\cap S(A) \neq \emptyset\}$. Suppose $\mu\BSP\nu$. Then for any $M\in S/{\BSP}$, $\mu(M) = \nu(M)$, hence $\mu(A) = \mu(S(A)) = \nu(S(A)) = \nu(A).$ 

Let $\TRANA{\mu}{a}{\mu'}$. Then for any $s\in S$ there exists $\TRANPA{s}{a}{\mu_s}$ such that
$$\mu' = \sum_{s\in S} \mu(s)\cdot \mu_s.$$
Now for each $t\in S$, let $[t]_{\BSP}$ be the equivalence class of $\BSP$ which contains $t$. Then for every $s \in [t]_{\BSP}$, to match the transition $\TRANPA{s}{a}{\mu_s}$ there exists some $\nu_t^s$ such that $\TRANPA{t}{a}{\nu^s_t}$ and $\mu_s\BSP \nu^s_t$. Let 
$$\nu_t = \sum_{s\in [t]_{\BSP}}\frac{\mu(s)}{\mu([t]_{\BSP})}\cdot \nu_t^s.$$
Then we have $\TRANPA{t}{a}{\nu_t}$, and $\TRANA{\nu}{a}{\nu'}$ where 
$$\nu' = \sum_{t\in S} \nu(t) \cdot\nu_t.$$
It remains to prove $\mu' \BSP \nu'$. For any $M\in S/{\BSP}$, since $\mu_s\BSP \nu^s_t$ we have 
\begin{eqnarray*}
\nu_t(M) &=& \sum_{s\in [t]_{\BSP}}\frac{\mu(s)}{\mu([t]_{\BSP})} \nu_t^s(M)= \sum_{s\in [t]_{\BSP}}\frac{\mu(s)}{\mu([t]_{\BSP})} \mu_s(M).
\end{eqnarray*} 
Thus
\begin{eqnarray*}
\nu'(M) &=&  \sum_{t\in S} \nu(t) \sum_{s\in [t]_{\BSP}}\frac{\mu(s)}{\mu([t]_{\BSP})} \mu_s(M)\\
&=& \sum_{s\in S}  \mu(s)\mu_s(M)\sum_{t\in [s]_{\BSP}}\frac{\nu(t)}{\nu([s]_{\BSP})}\\
&=& \sum_{s\in S} \mu(s) \mu_s(M)  = \mu'(M)
\end{eqnarray*} 
where for the second equality we have used the fact that $\mu([t]_{\BSP}) = \nu([s]_{\BSP})$ whenever $s\BSP t$. \qed
\end{proof}

Probabilistic bisimulation is defined over distributions inside one
automaton, whereas equivalence for reactive automata is
defined over two automata. However, they can be connected by the notion of
direct sum of two automata, which is the automaton obtained by
considering the disjoint union of states, edges and labelling functions
respectively.

\begin{lemma}\label{lem:eqra}
  Let $\A_1$ and $\A_2$ be two reactive automata with the same set of
  actions $Act$, and $\alpha_1$ and $\alpha_2$ the corresponding initial distributions. Then the following are equivalent:
  \begin{enumerate}
  \item $\A_1$ and
  $\A_2$ are equivalent,
\item $\alpha_1\sim_d\alpha_2$, 
\item $\alpha_1\sim \alpha_2$ in their direct sum. 
  \end{enumerate}
\end{lemma}
\begin{proof}
  The equivalence between (1) and (2) is shown in~\cite{DoyenHR08}. 
  The equivalence between (2) and (3) is straightforward, as for
  reactive automata our definition degenerates to Definition \ref{def:distribution_based}. \qed
\end{proof}

To conclude this subsection, we present an example to show that our bisimilarity is \emph{strictly} weaker than
$\BSP$.
\begin{example}\label{exa:lmc}
  Consider the example probabilistic automaton depicted in
  Fig.~\ref{fig:exam1}, which is inspired from an example
  in~\cite{DoyenHR08}. 
 Let $AP=Act=\{a\}$, $L(s) =\enableAct(s)$ for each $s$, and $\epsilon_1=\epsilon_2=0$. We argue that
  $q\not\BSP q'$. Otherwise, note $\TRANA{q}{a}{\frac12 \dirac{r_1} +
    \frac12 \dirac{r_2}}$ and $\TRANA{q'}{a}{\dirac{r'}}$. Then we must have $r' \BSP r_1
  \BSP r_2$. This is impossible, as $\TRANA{r_1}{a}{\frac23
    \dirac{s_1} + \frac13 \dirac{s_2}}$ and $\TRANA{r'}{a}{\frac12
    \dirac{s_1'} + \frac12 \dirac{s_2'}}$, but $s_1\BSP s_1'\not\BSP s_2\BSP s_2'$.

However, by our definition of bisimulation, the Dirac distributions
$\dirac{q}$ and $\dirac{q'}$ are indeed bisimilar. The reason is that we have the following transition
$$\TRANA{\frac12 \dirac{r_1} + \frac12\dirac{r_2}}{a}{\frac13 \dirac{s_1} + \frac16 \dirac{s_2} + \frac16 \dirac{s_3} + \frac13 \dirac{s_4}},$$
and it is easy to check $\dirac{s_1}\sim \dirac{s_3}\sim \dirac{s_1'}$ and $\dirac{s_2}\sim \dirac{s_4}\sim \dirac{s_2'}$. Thus we have $\frac12 \dirac{r_1} + \frac12 \dirac{r_2} \sim \dirac{r'}$, and finally $\dirac{q}\sim \dirac{q'}$.
\begin{figure}[!t]
  \centering
  \scalebox{0.6}{
\begin{tikzpicture}[->,>=stealth,auto,node distance=2cm,semithick,scale=1,every node/.style={scale=1}]
	\tikzstyle{state}=[minimum size=25pt,circle,draw,thick]
        \tikzstyle{triangleState}=[minimum size=0pt,regular polygon, regular polygon sides=4,draw,thick]
	\tikzstyle{stateNframe}=[]
	every label/.style=draw
        \tikzstyle{blackdot}=[circle,fill=black, minimum
        size=6pt,inner sep=0pt]
      \node[state](q){$q$};
      \node[blackdot,yshift=0.5cm](m1)[below of=q]{};
      \node[state,yshift=-0.5cm](r1)[below left of=m1,xshift=-1cm]{$r_1$};
      \node[state,yshift=-0.5cm](r2)[below right of=m1,xshift=1cm]{$r_2$};
      \node[blackdot,yshift=0.5cm](m2)[below of=r1]{};
      \node[blackdot,yshift=0.5cm](m3)[below of=r2]{};
      \node[state,yshift=-0.5cm](s1)[below left of=m2]{$s_1$};
      \node[state,yshift=-0.5cm](s2)[below right of=m2]{$s_2$};
      \node[state,yshift=-0.5cm](s3)[below left of=m3]{$s_3$};
      \node[state,yshift=-0.5cm](s4)[below right of=m3]{$s_4$};
      \node[state](qa)[left of=q,xshift=10cm]{$q'$};
      \node[state](ra)[below of=qa,yshift=-0.5cm]{$r'$};
      \node[blackdot,yshift=0.5cm](m4)[below of=ra,yshift=-0.5cm]{};
      \node[state,yshift=-0.5cm](s1a)[below left of=m4,yshift=-0.3cm]{$s'_1$};
      \node[state,yshift=-0.5cm](s2a)[below right of=m4,yshift=-0.3cm]{$s'_2$};
     \path (q) edge[-]     node[right] {$a$} (m1)
	   (m1) edge[dashed]     node[left,yshift=0.3cm] {$\frac{1}{2}$}   (r1)
                edge[dashed]     node[right,yshift=0.3cm] {$\frac{1}{2}$}   (r2)
           (r1) edge[-]             node {$a$} (m2)
           (r2) edge[-]             node {$a$} (m3)
           (m2) edge[dashed]     node[left] {$\frac{2}{3}+\epsilon_1$}   (s1)
                edge[dashed]     node[right] {$\frac{1}{3}-\epsilon_1$}   (s2)
           (m3) edge[dashed]     node[left] {$\frac{1}{3}-\epsilon_2$}   (s3)
                edge[dashed]     node[right] {$\frac{2}{3}+\epsilon_2$}   (s4)
           (s1) edge[loop left] node {$a$} (s1)
           (s3) edge[loop left] node {$a$} (s3)
           (qa) edge             node {$a$} (ra)
           (ra) edge[-]             node {$a$} (m4)
           (m4) edge[dashed]             node[left,yshift=0.1cm] {$\frac{1}{2}$} (s1a)
                edge[dashed]             node[right,yshift=0.1cm] {$\frac{1}{2}$} (s2a)
           (s1a)edge[loop left] node {$a$} (s1a);

\end{tikzpicture}
}
\caption{\label{fig:exam1}An illustrating example in which $L(s)=\enableAct(s)$ for each $s$.}
\end{figure}
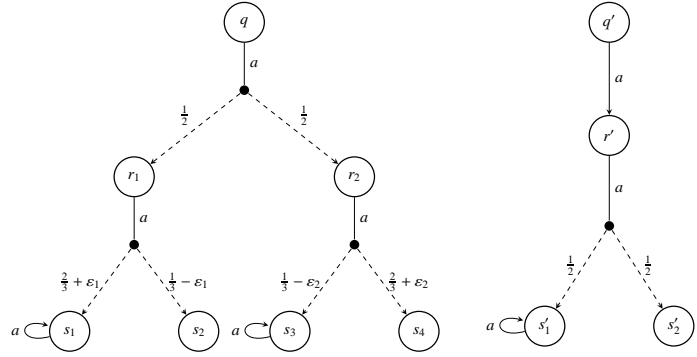
\end{example}

\subsection{Properties of the Relations}\label{sec:property}
In the following, we show that the notion of bisimilarity is in harmony with the linear combination and the limit of distributions. 

\begin{definition}\label{def:decom}
 A binary relation $R \subseteq
  \dist(S) \times\dist(S)$
 is said to be 
 \begin{itemize}
 \item \emph{linear}, if for any finite set $I$ and any probabilistic distribution $\{p_i\}_{i\in I}$, $\mu_i R \nu_i$ for each $i$
 implies
 $(\sum_{i\in I}p_i\cdot\mu_i)R(\sum_{i\in I}p_i\cdot\nu_i)$;
 \item \emph{continuous}, if for any convergent sequences of distributions $\{\mu_i\}_i$ and $\{\nu_i\}_i$, $\mu_i R \nu_i$ for each $i$ implies $(\lim_i \mu_i) R (\lim_i\nu_i)$;
 \item \emph{left-decomposable}, if
 $(\sum_{i\in I}p_i\cdot\mu_i)R\nu$, where $0< p_i \leq 1$ and $\sum_{i\in I}p_i = 1$, then $\nu$ can be written as $\sum_{i\in I}p_i\cdot\nu_i$
    such that $\mu_i R\nu_i$ for every $i \in I$.
\item \emph{left-convergent}, if
 $(\lim_i \mu_i)R\nu$, then for any $i$ we have $\mu_i R\nu_i$ for
 some $\nu_i$ with $\lim_i \nu_i = \nu$.
\end{itemize}
\end{definition}

We prove below that our transition relation between distributions satisfies these properties.
 
\begin{lemma}\label{lem:tranld} For an input-enabled probabilistic automaton, the transition relation $\TRANA{}{a}{}$ between distributions is linear, continuous, left-decomposable,
and left-convergent.
\end{lemma}
\begin{proof}
 \begin{itemize}
 \item \emph{Linearity}. Let $I$ be a finite index set and $\{p_i\mid i\in I\}$ a probabilistic distribution on $I$. 
Suppose $\TRANA{\mu_i}{a}{\nu_i}$ for each $i \in I$.  Then by definition, for each
   $s$ there exists
   $\TRANPA{s}{a}{\mu^i_s}$ such that 
   $\nu_i=\sum_{s}\mu_i(s)\cdot\mu^i_s$.
Now let $\mu=\sum_{i\in I}p_i\cdot\mu_i$. Then for each $s\in \supp(\mu)$,
$$\TRANPA{s}{a}\mu_s:=\sum_{i\in I}{\frac{p_i\mu_i(s)}{\mu(s)}\cdot \mu^i_s}.$$
 On the other hand, we check that 
\begin{eqnarray*}\nu := \sum_{i\in I} p_i \cdot \nu_i &=& 
\sum\limits_{s\in S}\sum_{i\in I} p_i\mu_i(s)\cdot\mu^i_s = \sum\limits_{s\in S }\mu(s)\cdot\mu_s .
\end{eqnarray*}
 Thus $\TRANA{\mu}{a}{\nu}$ as expected.
 
\item \emph{Continuity}. Suppose $\TRANA{\mu_i}{a}{\nu_i}$ for each $i \in I$, and $\lim_i \mu_i = \mu$.  By definition, for each
   $s$ there exists
   $\TRANPA{s}{a}{\mu^i_s}$ such that 
   $\nu_i=\sum_{s}\mu_i(s)\cdot\mu^i_s$.
Note that $\dist(S)$ is a compact set. For each $s$ we can choose a convergent subsequence $\{\mu_s^{i_k}\}_k$ of $\{\mu_s^{i}\}_i$ such that
$\lim_k \mu_s^{i_k} = \mu_s$ for some $\mu_s$. Then $\TRANPA{s}{a}{\mu_s}$, and 
\begin{eqnarray*}\TRANA{\mu}{a}{\nu} :=  \sum_{s\in S}\mu(s)\cdot\mu_s.
\end{eqnarray*}
Note that for each $k$,
 \begin{eqnarray*} \|\nu_{i_k} - \nu\|_1 &\leq & \|\mu_{i_k} - \mu\|_1 + \sum\limits_{s\in S} \mu(s) \|\mu^{i_k}_s - \mu_s\|_1
\end{eqnarray*}
where $\|\cdot\|_1$ denotes the $l_1$-norm.
We have $\nu=\lim_k\nu_{i_k}$ by the assumption that $\lim_i \mu_i = \mu$. 
Thus $\lim_i \nu_i = \nu$, as $\{\nu_i\}_i$ itself converges.

 \item \emph{Left-decomposability}. Let
 $\mu:=\TRANA{(\sum_{i\in I}p_i\cdot\mu_i)}{a}{\nu}$. Then by definition, for each
   $s$ there exists
   $\TRANPA{s}{a}{\mu_s}$ such that 
   $\nu=\sum_{s}\mu(s)\cdot\mu_s$. Thus 
$$\TRANA{\mu_i}{a}\nu_i:=\sum\limits_{s\in S }\mu_i(s)\cdot\mu_s.$$
Finally, it is easy to show that $\sum_{i\in I} p_i \cdot \nu_i = \nu$.  
 
 \item \emph{Left-convergence}. Similar to the last case.   \qed
\end{itemize}
\end{proof}
%
%
%

\begin{theorem}\label{thm:ld}
The bisimilarity relation $\sim$ is both linear and continuous.
\end{theorem}
\begin{proof}
Note that if $\mu_i\in \dist(S)$ for any $i$, then both $\sum_i p_i\cdot \mu_i$ and $\lim_i \mu_i$ (if exists) are again in $\dist(S)$.
Thus we need only consider the case when the automaton is input-enabled.
\begin{itemize}
\item
 \emph{Linearity}. It suffices to show that the symmetric relation 
 \begin{eqnarray*}
 R&=&\left\{\left(\sum_{i\in I} p_i\cdot \mu_i, \sum_{i\in I} p_i\cdot \nu_i\right) \mid I \mbox{
finite}, \sum_{i\in I} p_i =1, \forall i.(p_i\geq0 \wedge \mu_i \sim\nu_i)\right\}
 \end{eqnarray*}
is a bisimulation. Let $\mu=\sum_{i\in I}p_i\cdot\mu_i$, $\nu=\sum_{i\in I}p_i\cdot\nu_i$, and $\mu R\nu$. 
 Then for any $A\subseteq AP$, 
 $$\mu(A) = \sum_{i\in I} p_i\cdot \mu_i(A) = \sum_{i\in I} p_i\cdot \nu_i(A)
  = \nu(A).$$
 Now suppose $\TRANA{\mu}{a}{\mu'}$. 
 Then by Lemma~\ref{lem:tranld} (left-decomposability),  for each $i\in I$ we have
 $\TRANA{\mu_i}{a}{\mu'_i}$ for some $\mu'_i$ such that $\mu' = \sum_{i}p_i\cdot \mu'_i$. From the assumption that
 $\mu_i \sim\nu_i$, we derive $\TRANA{\nu_i}{a}{\nu'_i}$ with $\mu'_i \sim\nu'_i$ for each $i$. Thus $\TRANA{\nu}{a}{\nu'}
 :=\sum_{i}p_i\cdot \nu'_i$ by Lemma~\ref{lem:tranld} again (linearity). Finally, it is obvious that $(\mu', \nu')\in R$. 
 
 \item \emph{Continuity}.  It suffices to show that the symmetric relation
\[ R=\{ (\mu, \nu) \mid \forall  i\geq 1, \mu_i\sim{\nu_i}, \lim_i \mu_i = \mu, \mbox{ and }\lim_i \nu_i = \nu\}\]
is a bisimulation. First, for any $A\subseteq AP$, we have
\[\mu(A) = \lim_i \mu_i(A) = \lim_i \nu_i(A) = \nu(A). \]
Let $\TRANA{\mu}{a}{\mu'}$. By Lemma~\ref{lem:tranld} (left-convergence),  for any $i$ we have
$\TRANA{\mu_i}{a}{\mu_i'}$ with $\lim_{i} \mu_i' =\mu'.$
 To match the transitions, we have $\TRANA{\nu_i}{a}{\nu'_i}$ such that $\mu_i'\sim{\nu'_i}$.
Note that $\dist(S)$ is a compact set. We can choose a convergent subsequence $\{\nu_{i_k}'\}_k$ of $\{\nu_{i}'\}_i$ such that
$\lim_k \nu_{i_k}' = \nu'$ for some $\nu'$. From the fact that $\lim_i \nu_i = \nu$ and Lemma~\ref{lem:tranld} (continuity), it holds 
$\TRANA{\nu}{a}{\nu'}$ as well. Finally, it is easy to see that $(\mu', \nu')\in R$. 
 \end{itemize}
 \qed
\end{proof}

In general, our definition of bisimilarity is not
left-decomposable. This is in sharp contrast with the 
bisimulations defined by using the lifting technique~\cite{DengGHM09}. However, it
should not be regarded as a shortcoming; actually it is the key
requirement we abandon in this paper, which makes our definition
reasonably weak. This has been clearly illustrated in 
Example~\ref{exa:lmc}.

\section{Bisimulation Metrics}\label{sec:metric}
We present distribution-based bisimulation metrics with discounting factor $\gamma\in (0,1]$ in
this section.  Three different ways of defining bisimulation
metrics between states exist in the literature: one coinductive
definition based on bisimulations~\cite{YingW00,Ying01,Ying02,DesharnaisLT08}, one based on
the maximal logical differences~\cite{DesharnaisGJP99,DesharnaisGJP04,BreugelSW07},
and one on fixed point~\cite{AlfaroMRS07,BreugelSW07,FernsPP11}. We propose all the three
versions for our distribution-based bisimulations with discounting. Moreover, we show that they coincide.  We fix a discounting
factor $\gamma\in (0,1]$ throughout this section. For any $\mu, \nu\in \dist(S)$, we define the distance
$$d_{AP}(\mu, \nu) := \frac 12 \sum_{A\subseteq AP} \left|\mu(A) - \nu(A)\right|.$$
Then it is easy to check that $$d_{AP}(\mu, \nu) = \max_{\B\subseteq 2^{AP}} \left|\sum_{A\in \B} \mu(A) - \sum_{A\in \B} \nu(A)\right|= \max_{\B\subseteq 2^{AP}} \left[\sum_{A\in \B} \mu(A) - \sum_{A\in \B} \nu(A)\right].$$
\subsection{A Direct Approach}

\begin{definition}\label{def:bisi-metric}
Let $\A=(S, Act, \rightarrow,L,\alpha)$ be an input-enabled probabilistic automaton.  A family of symmetric relations $\{R_\epsilon \mid \epsilon \geq 0\}$ over $\dist(S)$ is a (discounted) approximate bisimulation if
for any $\epsilon\geq 0$ and $\mu R_\epsilon\nu$, we have 
  \begin{enumerate}
\item $d_{AP}(\mu, \nu)\leq \epsilon$;
\item 
for each $a\in Act$,  $\TRANA{\mu}{a}{\mu'}$ implies that there exists a transition
  $\TRANA{\nu}{a}{\nu'}$ such that $\mu' R_{\epsilon/\gamma}\nu'$.
\end{enumerate}
We write $\mu\sim_\epsilon^\A\nu$ whenever there is an approximate
 bisimulation $\{R_\epsilon \mid \epsilon\geq 0\}$ such that $\mu R_\epsilon \nu$.  For any two
distributions $\mu$ and $\nu$, we define the bisimulation distance
of $\mu$ and $\nu$ as

\begin{equation}\label{def:metric}
D_b^\A(\mu, \nu) = \inf\{\epsilon \geq 0 \mid \mu\sim^\A_{\epsilon} \nu\}.
\end{equation}
\end{definition}

Again, the approximate bisimulation and bisimulation distance of distributions in a general probabilistic automaton can be
defined in terms of the corresponding notions in the input-enabled extension; that is, $\mu\sim_\epsilon^\A\nu$ if $\mu\sim_\epsilon^{\A_\bot}\nu$, and $D_b^\A(\mu, \nu) := D_b^{\A_\bot}(\mu, \nu)$. We always omit the superscripts for simplicity if no confusion arises.

It is standard to show that the family $\{\sim_\epsilon \mid \epsilon\geq 0\}$ is itself an approximate bisimulation.
The following lemma collects some properties of $\sim_\epsilon$.

\begin{lemma}\label{lem:abis}
\begin{enumerate}
\item For each $\epsilon$, the $\epsilon$-bisimilarity $\sim_\epsilon$ is both
linear and continuous. 
\item If $\mu\sim_{\epsilon_1}\nu$ and $\nu\sim_{\epsilon_2}\omega$, then $\mu\sim_{\epsilon_1+\epsilon_2}\omega$;
\item $\sim_{\epsilon_1}{\subseteq}\sim_{\epsilon_2}$ whenever $\epsilon_{1}\leq \epsilon_{2}$.
\end{enumerate}
\end{lemma}
\begin{proof} 
The proof of item 1 is similar to
Theorem~\ref{thm:ld}.
For item 2, it suffices to show that $\{R_\epsilon \mid \epsilon\geq 0\}$ where $R_\epsilon\ {=} \bigcup_{\epsilon_1+\epsilon_2 = \epsilon}\left(\sim_{\epsilon_1}\circ\sim_{\epsilon_2}\right)$ is an approximate bisimulation (in the extended automaton, if necessary), which is routine. For item 3, suppose $\epsilon_2>0$. It is  easy to show $\{R_\epsilon \mid \epsilon\geq 0\}$,  $R_\epsilon\ {=} \sim_{\epsilon\epsilon_1/\epsilon_2}$, is an approximate bisimulation. 
Then if  $\mu\sim_{\epsilon_1}\nu$, that is, $\mu\sim_{\epsilon_2\epsilon_1/\epsilon_2}\nu$, we have $\mu R_{\epsilon_2}\nu$, and thus $\mu\sim_{\epsilon_2}\nu$ as required. \qed
\end{proof}

The following theorem states that the infimum in the definition
Eq.~\eqref{def:metric} of bisimulation distance can be replaced by
minimum; that is, the infimum is achievable.
\begin{theorem}\label{thm:metric}
For any $\mu,\nu\in \dist(S)$,  $\mu\sim_{D_b(\mu,\nu)} \nu.$
\end{theorem}
\begin{proof} By definition, we need to prove $\mu\sim_{D_b(\mu,\nu)} \nu$ in the extended automaton. We first prove that for any $\epsilon\geq0$, the symmetric relations $\{R_\epsilon \mid \epsilon \geq 0\}$ where 
\begin{eqnarray*}
R_\epsilon=\{(\mu, \nu)& \mid &\mu \sim_{\epsilon_{i}} \nu \mbox{ for each } \epsilon_{1}\geq\epsilon_{2}\geq\cdots \geq 0, \mbox{ and } \lim_{i\rightarrow \infty} \epsilon_{i}=\epsilon\}
\end{eqnarray*}
is an approximate bisimulation. Suppose $\mu R_\epsilon\nu$. Since $\mu\sim_{\epsilon_{i}} \nu$ we have $d_{AP}(\mu, \nu)\leq \epsilon_i$ for each $i$. Thus $d_{AP}(\mu, \nu)\leq \epsilon$ as well. 
Furthermore, if $\TRANA{\mu}{a}{\mu'}$, then for any $i\geq 1$,  $\TRANA{\nu}{a}{\nu_{i}}$ and
$\mu'\sim_{\epsilon_{i}/\gamma}\nu_{i}$. Since $\dist(S)$ is compact, there exists a subsequence $\{\nu_{i_k}\}_k$ of $\{\nu_i\}_i$ such that $\lim_k \nu_{i_k} = \nu'$ for some $\nu'$. We claim that
\begin{itemize}
\item $\TRANA{\nu}{a}{\nu'}$. This follows from the continuity of the transition $\TRANA{}{a}{}$, Lemma~\ref{lem:tranld}.
\item For each $k\geq 1$,  $\mu'\sim_{\epsilon_{i_k}/\gamma}\nu'$. Suppose conversely that $\mu'\not\sim_{\epsilon_{i_k}/\gamma}\nu'$ for some $k$. Then by the continuity of $\sim_{\epsilon_{i_k}/\gamma}$, we have $\mu'\not\sim_{\epsilon_{i_k}/\gamma}\nu_j$ for some $j\geq i_k$. This contradicts the fact that  $\mu'\sim_{\epsilon_j/\gamma}\nu_{j}$ and Lemma~\ref{lem:abis}(3). Thus $\mu' R_{\epsilon/\gamma}\nu'$ as required.
\end{itemize}

Finally, it is direct from definition that there exists a decreasing sequence $\{\epsilon_i\}_i$ such that $\lim_i \epsilon_i = D_b(\mu, \nu)$ and $\mu \sim_{\epsilon_{i}} \nu$ for each $i$. Then the theorem follows.
\qed
\end{proof}

A direct consequence of the above theorem is that the bisimulation distance between two distributions vanishes if and only if they are  bisimilar.
\begin{corollary} For any $\mu, \nu\in \dist(S)$, $\mu\sim \nu$ if and only if $D_b(\mu,\nu)=0$. \end{corollary}
\begin{proof} Direct from Theorem~\ref{thm:metric}, by noting that $\sim{=}\sim_{0}$. \qed\end{proof}

The next theorem shows that $D_b$ is indeed a pseudo-metric.
\begin{theorem}
The bisimulation distance $D_b$ is a pseudo-metric on $\dist(S)$.
\end{theorem}
\begin{proof} We need only to prove that $D_b$ satisfies the triangle inequality
$$D_b(\mu,\nu) + D_b(\nu,\omega) \geq D_b(\mu,\omega).$$
By Theorem~\ref{thm:metric}, we have $\mu\sim_{D_b(\mu,\nu)}\nu$ and $\nu\sim_{D_b(\nu,\omega)}\omega$. Then the result follows from Lemma~\ref{lem:abis}(2). \qed
\end{proof}

\subsection{Modal Characterization of the Bisimulation Metric}
We now present a Hennessy-Milner type modal logic motivated by~\cite{DesharnaisGJP99,DesharnaisGJP04} to characterize the distance between distributions.

\begin{definition} The class $\l_m$ of modal formulae over $AP$, ranged over by $\phi$, $\phi_1$, $\phi_2$, etc, is defined by the following grammar:
\begin{eqnarray*}
\phi &::=& \B \ |\ \phi\oplus p\ |\ \neg \phi\ |\  \bigwedge_{i\in I} \phi_i\ |\ \<a\>\phi
\end{eqnarray*}
where $\B\subseteq 2^{AP}$, $p\in [0,1]$, $a\in Act$, and $I$ is an index set.
\end{definition}

Given an input-enabled probabilistic automaton $\A=(S, Act, \rightarrow,L,\alpha)$ over $AP$,
instead of defining the satisfaction relation $\models$ in the
qualitative setting, the
(discounted) semantics of the logic $\l_m$ is given in terms of
functions from $\dist(S)$ to $[0,1]$.  For any formula $\phi\in \l_m$,
the satisfaction function of $\phi$, denoted by $\phi$ again for
simplicity, is defined in a structurally inductive way as follows:
\begin{itemize}
\item $\B(\mu) := \sum_{A\in \B}\mu(A)$;
\item $(\phi\oplus p)(\mu) := \min\{\phi(\mu)+p, 1\}$;
\item $(\neg \phi)(\mu) := 1-\phi(\mu)$;
\item $ (\bigwedge_{i\in I} \phi_i)(\mu) := \inf_{i\in I} \phi_{i}(\mu)$;
\item $(\<a\>\phi)(\mu) := \sup_{\TRANA{\mu}{a}{\mu'}} \gamma \cdot\phi(\mu')$.
\end{itemize}


\begin{lemma}\label{lem:continuous}
For any $\phi\in \l_m$, $\phi : \dist(S) \rightarrow [0,1]$ is a continuous function.
\end{lemma}
\begin{proof} We prove by induction on the structure of $\phi$. The basis case when $\phi \equiv \B$ is obvious. The case of $\phi \equiv \phi'\oplus p$, $\phi \equiv \neg \phi'$, and $\phi \equiv \bigwedge_{i\in I} \phi_i$ are all easy from induction.
In the following we only consider the case when $\phi \equiv \<a\>\phi'$. 

Take arbitrarily $\{\mu_i\}_i$ with $\lim_i \mu_i = \mu$. We need to show there exists a subsequence $\{\mu_{i_k}\}_k$ of $\{\mu_i\}_i$ such that $\lim_k\phi(\mu_{i_k}) = \phi(\mu)$. Take arbitrarily $\epsilon >0$. 
\begin{itemize}
\item Let $\mu^*\in \dist(S)$ such that $\TRANA{\mu}{a}{\mu^*}$ and $\phi(\mu)\leq \gamma \cdot\phi'(\mu^*)+\epsilon/2$. We have from the left-convergence of $\TRANA{}{a}{}$ that $\TRANA{\mu_i}{a}{\nu_i}$ for some $\nu_i$, and $\lim_i \nu_i = \mu^*$. By induction, $\phi'$ is a continuous function. Thus we can find $N_1\geq 1$ such that for any $i\geq N_1$, $|\phi'(\mu^*)-\phi'(\nu_i)|<\epsilon/2\gamma$.

\item For each $i\geq 1$, let $\mu_i^*\in \dist(S)$ such that $\TRANA{\mu_i}{a}{\mu_i^*}$ and $\phi(\mu_i)\leq \gamma \cdot\phi'(\mu_i^*)+\epsilon/2$. Then we have $\TRANA{\mu}{a}{\nu^*}$ with $\nu^*=\lim_k \mu_{i_k}^*$ for some convergent subsequence $\{\mu_{i_k}^*\}_k$ of $\{\mu_i^*\}_i$. Again, from the induction that $\phi'$ is continuous, we can find $N_2\geq 1$ such that for any $k\geq N_2$, $|\phi'(\mu_{i_k}^*)-\phi'(\nu^*)|<\epsilon/2\gamma $.
\end{itemize}
Let $N=\max\{N_1, N_2\}$. Then for any $k\geq N$, we have from  $\TRANA{\mu}{a}{\nu^*}$ that
\begin{eqnarray*}
\phi(\mu_{i_k}) - \phi(\mu) &\leq & \gamma [\phi'(\mu_{i_k}^*) -  \phi'(\nu^*)] + \gamma \cdot\phi'(\nu^*) - \phi(\mu) + \epsilon/2\\
&\leq& \gamma [\phi'(\mu_{i_k}^*) -  \phi'(\nu^*)]+ \epsilon/2 < \epsilon.
\end{eqnarray*}
Similarly, from $\TRANA{\mu_{i_k}}{a}{\nu_{i_k}}$ we have
\begin{eqnarray*}
\phi(\mu) - \phi(\mu_{i_k})  &\leq& \gamma [\phi'(\mu^*) - \phi'(\nu_{i_k})] + \gamma \cdot\phi'(\nu_{i_k}) - \phi(\mu_{i_k}) + \epsilon/2 \\
&\leq&\gamma [\phi'(\mu^*) - \phi'(\nu_{i_k})]+ \epsilon/2 < \epsilon.
\end{eqnarray*}
Thus $\lim_k\phi(\mu_{i_k}) = \phi(\mu)$ as required.
\qed
\end{proof}

From Lemma~\ref{lem:continuous}, and noting that the set $\{\mu' \mid
\TRANA{\mu}{a}{\mu'}\}$ is compact for each $\mu$ and $a$, the
supremum in the semantic definition of $\<a\>\phi$ can be replaced by
maximum; that is, $(\<a\>\phi)(\mu) =
\max_{\TRANA{\mu}{a}{\mu'}}\gamma \cdot \phi(\mu')$.
Now we define the logical distance for
distributions.
\begin{definition} The \emph{logic distance} of $\mu$ and $\nu$ in $\dist(S)$ of an input-enabled automaton is defined by 
  \begin{align}
    \label{eq:2}
D_l^\A(\mu, \nu) = \sup_{\phi \in \l_m} |\phi(\mu) - \phi(\nu)|\ .    
  \end{align}
The logic distance for a general probabilistic automaton can be
defined in terms of the input-enabled extension; that is, $D_l^\A(\mu, \nu) := D_l^{\A_\bot}(\mu, \nu)$. 
We always omit the superscripts for simplicity.
\end{definition}

Now we can show that the logic distance exactly coincides with bisimulation distance for any pair of distributions.

\begin{theorem}\label{thm:dbdl}
$D_b = D_l$. 
\end{theorem}
\begin{proof}
As both $D_b$ and $D_l$ are defined in terms of the input-enabled extension of automata, we only need to prove the result for
input-enabled case.
Let $\mu, \nu\in \dist(S)$. We first prove $D_b(\mu, \nu) \geq D_l(\mu, \nu)$. It suffices to show by structural induction that for any $\phi\in \l_m$, 
$|\phi(\mu) - \phi(\nu)|\leq D_b(\mu, \nu).$
There are five cases to consider.
\begin{itemize}
\item $\phi \equiv \B$ for some $\B\subseteq 2^{AP}$. Then $|\phi(\mu) - \phi(\nu)| = |\sum_{A\in \B} [\mu(A) -\nu(A)]|
 \leq d_{AP}(\mu, \nu) \leq D_b(\mu, \nu)$ by Theorem~\ref{thm:metric}.
\item $\phi \equiv  \phi'\oplus p$. Assume without loss of generality that $\phi'(\mu)\geq \phi'(\nu)$. Then $\phi(\mu)\geq \phi(\nu)$. By induction, we have $\phi'(\mu) - \phi'(\nu)\leq  D_b(\mu, \nu)$.
Thus
$$|\phi(\mu) - \phi(\nu)| = \min\{\phi'(\mu)+p, 1\} -  \min\{\phi'(\nu)+p, 1\}\leq  \phi'(\mu) - \phi'(\nu) \leq  D_b(\mu, \nu).$$

\item $\phi \equiv \neg \phi'$. By induction, we have $|\phi'(\mu) - \phi'(\nu)|\leq  D_b(\mu, \nu)$, thus
$|\phi(\mu) - \phi(\nu)|  = |1- \phi'(\mu) - 1 + \phi'(\nu)| \leq  D_b(\mu, \nu)$ as well.

\item $\phi \equiv   \bigwedge_{i\in I} \phi_i$. Assume $\phi(\mu)\geq  \phi(\nu)$. For any $\epsilon>0$, let $j\in I$ such that $\phi_j(\nu) \leq \phi(\nu)+\epsilon$. By induction, we have $|\phi_j(\mu) - \phi_j(\nu)|\leq  D_b(\mu, \nu)$. Then
$$|\phi(\mu) - \phi(\nu)|  \leq  \phi_{j}(\mu) - \phi_{j}(\nu) +\epsilon \leq  D_b(\mu, \nu) +\epsilon,$$
and $|\phi(\mu) - \phi(\nu)|  \leq D_b(\mu, \nu)$ from the arbitrariness of $\epsilon$.

\item $\phi \equiv \<a\>\phi'$. Assume $\phi(\mu)\geq \phi(\nu)$. 
Let $\mu'_*\in \dist(S)$ such that $\TRANA{\mu}{a}{\mu'_*}$ and $\gamma \cdot\phi'(\mu'_*) = \phi(\mu)$. From Theorem~\ref{thm:metric}, we have 
$\mu\sim_{D_b(\mu, \nu)}\nu$. Thus there exists $\nu'_*$ such that $\TRANA{\nu}{a}{\nu'_*}$ and $\mu'_*\sim_{D_b(\mu, \nu)/\gamma}\nu'_*$. Hence $\gamma \cdot D_b(\mu'_*, \nu'_*) \leq D_b(\mu, \nu)$, and
$$|\phi(\mu) - \phi(\nu)|  \leq  \gamma \cdot[\phi'(\mu'_*) - \phi'(\nu'_*)]\leq  \gamma \cdot D_b(\mu'_*, \nu'_*) \leq  D_b(\mu, \nu)$$
where the second inequality is from induction.
\end{itemize}

Now we turn to the proof of $D_b(\mu, \nu) \leq D_l(\mu, \nu)$. We
will achieve this by showing that the symmetric relations
$R_\epsilon=\{(\mu,\nu) \mid D_l(\mu, \nu) \leq \epsilon \},$ where $\epsilon\geq 0$, constitute an approximate bisimulation. Let $\mu R_\epsilon \nu$ for some $\epsilon \geq 0$. First, for any $\B\subseteq 2^{AP}$ we have 
$$\left|\sum_{A\in \B} \mu(A) - \sum_{A\in \B} \nu(A)\right| = |\B(\mu) - \B(\nu)| \leq D_l(\mu, \nu)\leq \epsilon.$$ 
Thus $d_{AP}(\mu, \nu) \leq \epsilon$ as well.
Now suppose $\TRANA{\mu}{a}{\mu'}$ for some $\mu'$. 
We have to show that there is some $\nu'$ with  $\TRANA{\nu}{a}{\nu'}$ and $D_l(\mu', \nu')\leq \epsilon/\gamma$. Consider the set
$$\k = \{\omega\in \dist(S) \mid \TRANA{\nu}{a}{\omega} \mbox{ and } D_l(\mu', \omega)> \epsilon/\gamma\}.$$
For each $\omega\in \k$, there must be some $\phi_{\omega}$ such that $|\phi_{\omega}(\mu') - \phi_{\omega}(\omega)|> \epsilon/\gamma$. 
As our logic includes the operator $\neg$, we can always assume that $\phi_{\omega}(\mu') > \phi_{\omega}(\omega)+\epsilon/\gamma$.
Let $p= \sup_{\omega\in \k} \phi_{\omega}(\mu')$. Let
$$\phi_{\omega}'=\phi_{\omega} \oplus [p - \phi_{\omega}(\mu')], \ \ \ \phi' = \bigwedge_{\omega\in \k} \phi'_{\omega},\ \ \  \mbox{  and  }\ \ \  \phi = \<a\>\phi'.$$
Then from the assumption that $D_l(\mu, \nu) \leq \epsilon$, we have $|\phi(\mu) - \phi(\nu)| \leq \epsilon$. Furthermore, we check that for any $\omega\in \k$,
$$\phi'_\omega(\mu')=\phi_{\omega}(\mu') \oplus [p - \phi_{\omega}(\mu')]=p.$$
Thus $\phi(\mu) \geq \gamma\cdot\phi'(\mu')=\gamma\cdot p$.

Let 
$\nu'$ be the distribution such that $\TRANA{\nu}{a}{\nu'}$ and $\phi(\nu) =\gamma\cdot \phi'(\nu')$. We are going to show that $\nu'\not\in \k$, and then $D_l(\mu', \nu')\leq \epsilon/\gamma$ as required. For this purpose, assume conversely that $\nu'\in \k$. 
Then
\begin{eqnarray*}
\phi(\nu) &=& \gamma\cdot\phi'(\nu')\leq \gamma\cdot\phi'_{\nu'}(\nu')\leq \gamma\cdot [\phi_{\nu'}(\nu')  +  p - \phi_{\nu'}(\mu')]\\
&<&\gamma\cdot p-\epsilon\leq \phi(\mu)-\epsilon,
\end{eqnarray*}
contradicting the fact that $|\phi(\mu) - \phi(\nu)| \leq \epsilon$.

We have proven that $\{R_\epsilon \mid \epsilon\geq 0\}$ is an approximate bisimulation. Thus $\mu\sim_\epsilon \nu$, and so $D_b(\mu, \nu) \leq \epsilon$, whenever $D_l(\mu, \nu) \leq \epsilon$. So we have $D_b(\mu, \nu) \leq D_l(\mu, \nu)$ from the arbitrariness of $\epsilon$. 
\qed
\end{proof}

\subsection{A Fixed Point-Based Approach}
In the following, we denote by $\mathcal{M}$ the set of pseudo-metrics
over $\dist(S)$.  
Denote by $\z$ the zero pseudo-metric which assigns 0 to each pair of distributions. 
For any $d, d'\in\mathcal{M}$, we write $d\leq d'$
if $d(\mu, \nu)\leq d'(\mu, \nu)$ for any $\mu$ and $\nu$. Obviously
$\leq$ is a partial order, and $(\mathcal{M}, \leq)$ is a complete
lattice.

\begin{definition}\label{def:metricfunc} Let $\A=(S, Act, \rightarrow,L,\alpha)$ be an input-enabled probabilistic automaton.  
We define the function $F:\mathcal{M}\to
  \mathcal{M}$ as follows. For any $\mu, \nu\in \dist(S)$, 
  \begin{eqnarray*}
 F(d)(\mu,\nu) = \max_{a\in Act} &\{& d_{AP}(\mu, \nu),
\sup_{\TRANA{\mu}{a}{\mu'}} \inf_{\TRANA{\nu}{a}{\nu'}} \gamma\cdot
  d(\mu',\nu'), \sup_{\TRANA{\nu}{a}{\nu'}}
  \inf_{\TRANA{\mu}{a}{\mu'}} \gamma\cdot d(\mu',\nu')\}.
\end{eqnarray*}
Then, $F$ is monotonic
with respect to $\leq$, and by Knaster-Tarski theorem, $F$ has a
least fixed point, denoted $D_f^\A$, given by 
\begin{align*}
D^\A_f = \bigvee_{n=0}^\infty
F^n(\z) \ .  
\end{align*}
Here $\bigvee$ means the supremum with respect to the order $\leq$.
\end{definition}

Once again, the fixed point-based distance for a general probabilistic automaton can be
defined in terms of the input-enabled extension; that is, $D_f^\A(\mu, \nu) := D_f^{\A_\bot}(\mu, \nu)$. We always omit the superscripts for simplicity.

Similar to
Lemma~\ref{lem:continuous}, we can show that the supremum
(resp. infimum) in Definition~\ref{def:metricfunc} can be replaced by
maximum (resp. minimum).  
Now  we show
that $D_f$ coincides with $D_b$.
\begin{theorem}\label{thm:dfdb}
$D_f= D_b$.  
\end{theorem}
As both $D_f$ and $D_b$ are defined in terms of the input-enabled extension of automata, we only need to prove Theorem~\ref{thm:dfdb} for
input-enabled case, which will be obtained by combining Lemma \ref{lem:left} and Lemma \ref{lem:right} below.
\begin{lemma}\label{lem:left}
For input-enabled probabilistic automata, $D_f\leq D_b$.
\end{lemma}
\begin{proof}
It suffices to prove by induction that for any $n\geq 0$, $F^n(\z)\leq D_b$. The case of $n=0$ is trivial. Suppose
$F^n(\z)\leq D_b$ for some $n\geq 0$. Then for any $a\in Act$ and any $\mu, \nu$, we have
\begin{enumerate}
\item[(1)] $d_{AP}(\mu, \nu) \leq D_b(\mu, \nu)$ by the fact that $\mu\sim_{D_b(\mu, \nu)} \nu$;
\item[(2)] Note that $\mu\sim_{D_b(\mu, \nu)} \nu$. Whenever $\TRANA{\mu}{a}{\mu'}$, we have $\TRANA{\nu}{a}{\nu'}$ for some $\nu'$ such that $\mu'\sim_{D_b(\mu, \nu)/\gamma} \nu'$, and hence $\gamma\cdot D_b(\mu', \nu')\leq D_b(\mu, \nu)$. That, together with the assumption $F^n(\z)\leq D_b$, implies
$$ \max_{\TRANA{\mu}{a}{\mu'}} \min_{\TRANA{\nu}{a}{\nu'}}\gamma\cdot  F^n(\z)(\mu',\nu') \leq D_b(\mu,\nu).$$
The symmetric form can be similarly proved.
\end{enumerate}
Summing up (1) and (2), we have $F^{n+1}(\z)\leq D_b$. \qed
\end{proof}

To prove the other direction, we first introduce the notion of bounded approximate bisimulations.
\begin{definition} Let $\A$ be an input-enabled probabilistic automaton.
We define symmetric relations
\begin{itemize}
\item $\ysim{\epsilon}_0 {:=} \dist(S)\times \dist(S)$ for any $\epsilon \geq 0;$
\item for $n\geq 0$, $\mu\ysim{\epsilon}_{n+1} \nu$ if  $d_{AP}(\mu, \nu) \leq \epsilon$ and whenever $\TRANA{\mu}{a}{\mu'}$, there exists $\TRANA{\nu}{a}{\nu'}$ for some $\nu'$ such that $\mu'\ysim{\epsilon/\gamma}_{n} \nu'$.
\item $\ysim{\epsilon} {:=} \bigcap_{n\geq 0} \ysim{\epsilon}_n$.
\end{itemize}
\end{definition}

The following lemma collects some useful properties of $\ysim{\epsilon}_n$ and $\ysim{\epsilon}$.

\begin{lemma}\label{lem:simlimit}
\begin{enumerate}
\item $\ysim{\epsilon}_n {\subseteq} \ysim{\epsilon}_m$ provided that $n \geq m$;
\item for any $n\geq 0$, $\ysim{\epsilon}_n {\subseteq} \ysim{\epsilon'}_n$ provided that $\epsilon \leq \epsilon'$;
\item for any $n\geq 0$, $\ysim{\epsilon}_n$ is continuous;
\item $\ysim{\epsilon} {=} \sim_\epsilon$.
\end{enumerate}
\end{lemma}
\begin{proof} Items 1, 2, and 3 are easy by induction, and so is the $ \sim_\epsilon {\subseteq} \ysim{\epsilon}$ part of item 4. To prove $\ysim{\epsilon} {\subseteq} \sim_\epsilon$, we show that $\{\ysim{\epsilon} \mid \epsilon\geq 0\}$ is an  approximate bisimulation. Suppose $\mu\ysim{\epsilon}\nu$. Then $d_{AP}(\mu, \nu) \leq \epsilon$ by definition. Now let $\TRANA{\mu}{a}{\mu'}$.
 For each $n\geq 0$, from the assumption that $\mu\ysim{\epsilon}_{n+1}\nu$ we have $\TRANA{\nu}{a}{\nu_n}$ such that  $\mu'\ysim{\epsilon/\gamma}_{n}\nu_n$. Let $\{\nu_{i_k}\}_k$ be a convergent subsequence of $\{\nu_n\}_n$ such that $\lim_k \nu_{i_k} = \nu'$ for some $\nu'$. Then from the continuity of $\TRANA{}{a}{}$ we have $\TRANA{\nu}{a}{\nu'}$. We claim further that $\mu'\ysim{\epsilon/\gamma} \nu'$. Otherwise there exists $N$ such that 
$\mu'\not\ysim{\epsilon/\gamma}_{N}\nu'$. Now by the continuity of $\ysim{\epsilon/\gamma}_{N}$, we have $\mu'\not\ysim{\epsilon/\gamma}_{N}\nu_j$ for some $j\geq N$. This contradicts the fact that  $\mu'\ysim{\epsilon/\gamma}_j\nu_{j}$ and item 1. \qed
\end{proof}

\begin{lemma}\label{lem:tmp}
For any $n\geq 0$, we have $\mu \ysim{F^n(\z)(\mu, \nu)}_n \nu$.
\end{lemma}
\begin{proof} 
We prove this lemma by induction on $n$. The case of $n=0$ is trivial. Suppose $\mu \ysim{F^n(\z)(\mu, \nu)}_n \nu$ for some $n\geq 0$. Let $a\in Act$. By definition, we have  
$$F^{n+1}(\z)(\mu,\nu) \geq \max_{\TRANA{\mu}{a}{\mu'}} \min_{\TRANA{\nu}{a}{\nu'}} \gamma\cdot F^n(\z)(\mu',\nu').
$$
Thus for any $\TRANA{\mu}{a}{\mu'}$, there exists $\TRANA{\nu}{a}{\nu'}$ such that $\gamma\cdot F^n(\z)(\mu',\nu')\leq F^{n+1}(\z)(\mu,\nu)$. By induction, we know $\mu' \ysim{F^n(\z)(\mu', \nu')}_n \nu'$, thus $\mu' \ysim{ F^{n+1}(\z)(\mu,\nu)/\gamma}_n \nu'$ from Lemma~\ref{lem:simlimit}(2). On the other hand, we have $F^{n+1}(\z)(\mu,\nu) \geq d_{AP}(\mu, \nu)$ by definition.
Thus we have $\mu \ysim{ F^{n+1}(\z)(\mu,\nu)}_{n+1} \nu$.
\qed
\end{proof}

With the two lemmas above, we can prove that $D_b\leq D_f$.
\begin{lemma}\label{lem:right}
For input-enabled probabilistic automata, $D_b\leq D_f$.
\end{lemma}
\begin{proof}
For any $\mu$ and $\nu$, by Lemmas~\ref{lem:tmp} and \ref{lem:simlimit}(2), we have $\mu \ysim{D_f(\mu, \nu)}_n \nu$ for all $n\geq 0$, so $\mu \ysim{D_f(\mu, \nu)} \nu$ by definition. Then from 
 Lemma~\ref{lem:simlimit}(4) we have
$\mu \sim_{D_f(\mu, \nu)} \nu$, hence $D_b(\mu, \nu)\leq D_f(\mu, \nu)$. \qed
\end{proof}

\subsection{Comparison with State-Based Metric}\label{app:comparison}
In this subsection, we show that our distribution-based bisimulation
metric is upper bounded by the state-based game bisimulation
metric~\cite{AlfaroMRS07} for MDPs. This game bisimulation metric is
particularly attractive as it preserves probabilistic reachability,
long-run, and discounted average behaviours \cite{ChatterjeeAMR10}. We first recall the
definition of state-based game bisimulation metric for MDPs in \cite{AlfaroMRS07}:

\begin{definition}
Given $\mu, \nu\in \dist(S)$, $\mu\otimes \nu$ is defined as the set of \emph{weight} functions $\lambda: S\times S\rightarrow [0,1]$ such that
for any $s,t\in S$, 
$$\sum_{s\in S} \lambda(s,t) = \nu(t)\ \ \ \mbox{ and  }\ \ \ \sum_{t\in S} \lambda(s,t) = \mu(s).$$
Given a metric $d$ defined on $S$, we lift it to $\dist(S)$ by defining
$$d(\mu, \nu) = \inf_{\lambda\in \mu\otimes \nu} \left(\sum_{s,t\in S} \lambda(s,t)\cdot d(s, t)\right).$$
\end{definition}

Actually the infimum in the above definition is attainable. 

\begin{definition} 
We define the function $f:\mathcal{M}\to
  \mathcal{M}$ as follows. For any $s, t \in S$, 
  \begin{eqnarray*}
 f(d)(s,t) = \max_{a\in Act}\left\{1-\delta_{L(s), L(t)}, \sup_{\TRANPA{s}{a}{\mu}} \inf_{\TRANPA{t}{a}{\nu}} 
  \gamma\cdot d(\mu,\nu), \sup_{\TRANPA{t}{a}{\nu}}
  \inf_{\TRANPA{s}{a}{\mu}}  \gamma\cdot d(\mu,\nu)\right\}
\end{eqnarray*}
where $\delta_{L(s), L(t)}=1$ if $L(s)=L(t)$, and 0 otherwise. We take $\inf \emptyset = 1$ and $\sup\emptyset = 0$.
Again, $f$ is monotonic
with respect to $\leq$, and by Knaster-Tarski theorem, $F$ has a
least fixed point, denoted $d_f$, given by 
\begin{align*}
d_f = \bigvee_{n=0}^\infty
f^n(\z) \ .  
\end{align*}

\end{definition}

Now we can prove the quantitative extension of Lemma~\ref{lem:bsp}. Without loss of any generality, we assume that $\A$ itself is input-enabled. Let $d_n = f^n(\z)$ and $D_n = F^n(\z)$ in Definition~\ref{def:metricfunc}.


\begin{lemma}\label{lem:tmp11}
For any $n\geq 1$, $d_{AP}(\mu, \nu )\leq d_n(\mu, \nu)$.
\end{lemma}
\begin{proof}
Let $\lambda$ be the weight function such that
$d_n(\mu, \nu) = \sum_{s,t\in S} \lambda(s,t)\cdot d_n(s, t)$. 
Since $d_n(s, t)\geq 1-\delta_{L(s), L(t)}$, we have
$$d_n(\mu, \nu)\geq 1- \sum_{s,t:L(s)= L(t)} \lambda(s,t).$$
On the other hand, for any $A\subseteq AP$, recall that $S(A)=\{s\in S \mid L(s) = A\}$. Then
\begin{eqnarray*}
\mu(A) - \nu(A) &=& \sum_{s\in S(A)} \mu(s) - \sum_{t\in S(A)} \nu(t) \\
&=&  \sum_{s\in S(A)}\sum_{t\not\in S(A)}\lambda(s,t) - \sum_{t\in S(A)}\sum_{s\not\in S(A)}\lambda(s,t).
\end{eqnarray*}
Let $\B\subseteq 2^{AP}$ such that
$d_{AP}(\mu, \nu) = \sum_{A\in \B} [\mu(A) - \nu(A)]$. Then
$$d_{AP}(\mu, \nu) \leq \sum_{A\in \B} \sum_{s\in S(A)}\sum_{t\not\in S(A)}\lambda(s,t) \leq \sum_{s,t:L(s)\neq L(t)} \lambda(s,t),$$
and the result follows. \qed
\end{proof}

\begin{theorem}\label{thm:less}
 Let $\A$ be a probabilistic automaton. Then $D_f \leq d_f$.
\end{theorem}
\begin{proof}
We prove by induction on $n$ that  $D_n(\mu, \nu) \leq d_n(\mu, \nu)$ for any $\mu, \nu\in \dist(S)$ and $n\geq 0$. The case $n=0$ is obvious. Suppose the result holds for some $n-1\geq 0$. 
Then from Lemma~\ref{lem:tmp11}, we need only to show that for any $\TRANA{\mu}{a}{\mu'}$ there exists $\TRANA{\nu}{a}{\nu'}$
such that $\gamma\cdot D_{n-1}(\mu', \nu') \leq d_n(\mu, \nu)$.

Let $\TRANA{\mu}{a}{\mu'}$. Then for any $s\in S$, $\TRANPA{s}{a}{\mu_s}$ with $\mu' = \sum_{s\in S} \mu (s)\cdot \mu_s$.
By definition of $d_n$, for any $t\in S$, we have $\TRANPA{t}{a}{\nu_t}$ such that $\gamma\cdot d_{n-1}(\mu_s, \nu_t) \leq d_n(s, t)$. Thus
$\TRANA{\nu}{a}{\nu'} := \sum_{t\in S} \nu(t)\cdot \nu_t$, and by induction, $D_{n-1}(\mu', \nu')\leq d_{n-1}(\mu', \nu')$. Now it suffices to prove  $\gamma\cdot d_{n-1}(\mu', \nu')\leq d_n(\mu, \nu)$.

Let $\lambda\in \mu\otimes \nu$ and $\gamma_{s,t}\in \mu_s\otimes \nu_t$ be the weight functions such that
$$d_n(\mu, \nu)=\sum_{s,t\in S} \lambda(s,t)\cdot d_n(s, t),\ \  d_{n-1}(\mu_s, \nu_t)=\sum_{u, v\in S} \gamma_{s,t}(u,v)\cdot d_{n-1}(u, v).$$
Then 
\begin{eqnarray*}
d_n(\mu, \nu)&\geq &\gamma\cdot  \sum_{s,t\in S} \lambda(s,t)\cdot d_{n-1}(\mu_s, \nu_t)\\
&=&\gamma\cdot \sum_{u, v\in S}\sum_{s,t\in S} \lambda(s,t)\gamma_{s,t}(u,v)\cdot d_{n-1}(u, v).
\end{eqnarray*}
We need to show that the function $\eta(u, v) :=\sum_{s,t\in S} \lambda(s,t)\gamma_{s,t}(u,v)$ is a weight function for $\mu'$ and $\nu'$. Indeed, it is easy to check that
\begin{eqnarray*}
\sum_u \eta(u, v) &=& \sum_{s,t\in S} \lambda(s,t)\sum_u\gamma_{s,t}(u,v)= \sum_{s,t\in S} \lambda(s,t)\nu_t(v) \\
&= & \sum_{t\in S} \nu(t)\nu_t(v) = \nu'(v).
\end{eqnarray*}
Similarly, we have $\sum_v \eta(u, v) = \mu'(u)$.
\qed
\end{proof}

\begin{example}
  Consider Fig.~\ref{fig:exam1}, and assume $\epsilon_1\ge
\epsilon_2\geq 0$ and $\gamma=1$. It is easy to check
that $D_f(\dirac{q},\dirac{q'}) = \frac12(\epsilon_1-\epsilon_2)$. 
However, according to the state-based bisimulation metric, $d_f(\dirac{r_1},\dirac{r'})=\frac16+\epsilon_1$ and 
$d_f(\dirac{r_2},\dirac{r'})=\frac16+\epsilon_2$. Thus $d_f(\dirac{q},\dirac{q'})=\frac16+\frac12(\epsilon_1+\epsilon_2)$. 
\end{example}

\subsection{Comparison with Equivalence Metric}

Note that we can easily extend the equivalence relation defined in Definition~\ref{def:distribution_based}
to a notion of equivalence metric:
\begin{definition}[Equivalence Metric]
 Let $\A_1$ and $\A_2$ be two reactive automata with the same set of
  actions $Act$. We say
  $\A_1$ and $\A_2$ are $\epsilon$-equivalent, denoted $\A_1
  \sim_\epsilon^d \A_2$, if 
  for any input word $w=a_1a_2\ldots a_n\in Act^*$, 
  $|\A_1(w)-\A_2(w)| \le \epsilon$. Furthermore, the equivalence distance
  between $\A_1$ and $\A_2$ is defined by $D_d(\A_1, \A_2):=\inf\{\epsilon\geq 0\mid \A_1
  \sim_\epsilon^d \A_2\}$.
\end{definition}

Now we show that  for
reactive automata, the equivalence metric 
$D_d$ coincide with our un-discounted bisimulation metric $D_b$, which  
may be regarded as a quantitative extension of Lemma~\ref{lem:eqra}.
\begin{proposition}
Let $\A_1$ and $\A_2$ be two reactive automata with the same set of
  actions $Act$. Let the discounting factor $\gamma =1$. Then
  $D_d(\A_1, \A_2) = D_b(\alpha_1, \alpha_2)$ where $D_b$ is defined in the direct sum of $\A_1$ and $\A_2$.
\end{proposition}
\begin{proof}
We first show that $D_d(\A_1, \A_2) \leq D_b(\alpha_1, \alpha_2)$.
For each input word $w=a_1a_2\ldots a_n$, it is easy to check that $ \A_i(w)=\phi(\alpha_i)$
where $\phi= \langle a_1\rangle \langle
a_2\rangle\ldots \langle a_n\rangle (F_1\cup F_2)$.
As we have shown that $D_b=D_l$, it holds $|\A_1(w)-\A_2(w)| \le D_b(\alpha_1, \alpha_2)$, and hence $\A_1
  \sim_{D_b(\alpha_1, \alpha_2)}^d \A_2$. Then $D_d(\A_1, \A_2) \leq D_b(\alpha_1, \alpha_2)$ by definition. 
  
Now we turn to the proof of $D_d(\A_1, \A_2) \geq D_b(\alpha_1, \alpha_2)$. First we show that
$$R_\epsilon=\{ (\mu, \nu) \mid \mu\in \dist(S_1), \nu\in \dist(S_2), \A_1^\mu \sim_\epsilon^d \A_2^\nu\}$$
is an approximate bisimulation. Here for a probabilistic automaton $\A$, we denote by $\A^\mu$ the automaton which
is the same as $\A$ except that the initial distribution is replaced by $\mu$. Let $\mu R_\epsilon\nu$.
Since $L(s) \in \{\emptyset, AP\}$ for all $s\in S_1\cup S_2$, we have $\mu(AP) + \mu(\emptyset) = \nu(AP) + \nu(\emptyset) = 1$.
Thus
$$d_{AP}(\mu, \nu) =  |\mu(AP) - \nu(AP)|=  |\mu(F_1) - \nu(F_2)|.$$ 
Note that $\mu(F_1)=\A_1^\mu(e)$  and $\nu(F_2)=\A_2^\nu(e)$,  where $e$
is the empty string. Then $d_{AP}(\mu, \nu)=|\A_1^\mu(e)-\A_2^\nu(e)|\leq \epsilon$.

Let $\TRANA{\mu}{a}{\mu'}$ and $\TRANA{\nu}{a}{\nu'}$. We need to show $\mu' R_\epsilon\nu'$, that is, 
$\A_1^{\mu'} \sim_\epsilon^d \A_2^{\nu'}$. For any $w\in Act^*$ and $i=1,2$, note that $\A_i^{\mu'}(w) = \A_i^{\mu}(aw)$. Then 
 $$|\A_1^{\mu'}(w)-\A_2^{\nu'}(w)| = |\A_1^{\mu}(aw)-\A_2^{\nu}(aw)|  \le \epsilon,$$
 and hence  $\A_1^{\mu'} \sim_\epsilon^d \A_2^{\nu'}$ as required. 
 
 Having proven that $R_\epsilon$ is an approximate bisimulation, we know $\A_1 \sim_\epsilon^d \A_2$ implies $\alpha_1\sim_\epsilon \alpha_2$.
 Thus $$D_d(\A_1, \A_2) =\inf\{\epsilon \mid \A_1 \sim_\epsilon^d \A_2\} \geq \inf\{\epsilon \mid\alpha_1\sim_\epsilon \alpha_2\}=D_b(\alpha_1, \alpha_2).$$\qed
 \end{proof}

\section{Comparison with Distribution-based Bisimulations in Literature}
\label{sec:relation-literature}
In this section, we review some distribution-based definitions of
bisimulation in the literature and discuss their relations. We first recall the definition in~\cite{EisentrautGHSZ12} except that we focus on
its strong counterpart. For this, we need some notations. Recall that
$\enableAct(s)$ denotes
the set of actions which can be performed in $s$. A distribution $\mu$
is \emph{consistent}, denoted $\consistDist{\mu}$, if
$\enableAct(s)=\enableAct(t)$ for any
$s,t\in\supp(\mu)$, i.e., all states in the support of $\mu$ have the
same set of enabled actions. In case $\mu$ is consistent, we also let
$\enableAct(\mu)=\enableAct(s)$ for some $s\in\supp(\mu)$.

\begin{definition}\label{def:late-bisimulation}
  Let $\A=(S, Act, \rightarrow,L,\alpha)$ be a probabilistic automaton. A symmetric
  relation $R \subseteq \dist(S)\times \dist(S)$ is a $\late$-bisimulation if $\mu{R}\nu$ implies that
  \begin{enumerate}
  \item $\mu(A)=\nu(A)$ for each $A\subseteq\AP$,
  \item for each $a\in Act$ whenever $\TRANA{\mu}{a}{\mu'}$,
    there exists a transition $\TRANA{\nu}{a}{\nu'}$ such that $\mu'{R}\nu'$, and
  \item if not $\consistDist{\mu}$, there exist convex decompositions $\mu=\sum_{1\le
      i\le n}p_i\cdot\mu_i$ and $\nu=\sum_{1\le i\le n}p_i\cdot\nu_i$
    such that $\consistDist{\mu_i}$ and $\mu_i{R}\nu_i$ for each $1\le
    i\le n$.
  \end{enumerate}
We write $\mu{\;\;\sim_{\late}^\A\;\;}\nu$ if there is a 
  $\late$-bisimulation $R$ such that $\mu{R}\nu$.
\end{definition}

Note that Definition~\ref{def:late-bisimulation} is given directly for general probabilistic automata without 
the need for input-enabled extension.
Recently, another definition of bisimulation based on distributions was
introduced in~\cite{HermannsKK14}. The main difference arises in the
lifting transition relation of distributions. Let $\actSet\subseteq
Act$ and $S_\actSet=\{s\in S\mid\enableAct(s)\cap \actSet\neq\emptyset\}$, i.e.,
$S_\actSet$ contains all states which is able to perform an action in $\actSet$. Instead of labelling a transition with a single
action, transitions of states and distributions in~\cite{HermannsKK14} are labelled by a
set of actions, denoted $\janTran{}{\actSet}{}$.
Formally, $\janTran{s}{\actSet}{\mu}$ if there exists $a\in \actSet$ such that
$\janTran{s}{a}{\mu}$, otherwise we say actions $\actSet$ are blocked at
$s$. Accordingly, $\janTran{\mu}{\actSet}{\nu}$ if $\mu(S_\actSet)>0$ and for each
$s\in S_\actSet\cap\supp(\mu)$, there exists $\janTran{s}{\actSet}{\mu_s}$ such
that $$\nu=\frac{1}{\mu(S_\actSet)}\sum_{s\in
  S_\actSet\cap\supp(\mu)}\mu(s)\cdot\mu_s.$$ Intuitively, a distribution
$\mu$ is able to perform a transition with label $\actSet$
if and only if at least one of its supports can perform an action in $\actSet$. Furthermore, all states in $S_\actSet\cap\supp(\mu)$ should
perform such a transition in the meanwhile. The resulting distribution
is the weighted sum of all the resulting distributions with weights
equal to their probabilities in $\mu$. Since it may happen that some states in
$\supp(\mu)$ cannot perform such a transition, i.e., $\supp(\mu)\not\subseteq S_{\actSet}$, we need the
normalizer $\frac{1}{\mu(S_\actSet)}$ in order to obtain a valid distribution.
Below follows the definition of bisimulation in~\cite{HermannsKK14},
where $\mu(\actSet,A)=\mu(\{s\in S_\actSet\mid L(s)=A\})$, the
probability of states in $\mu$ labelled by $A$ while being able to
perform actions in $\actSet$.
\begin{definition}\label{def:jan-bisimulation}
    Let $\A=(S, Act, \rightarrow,L,\alpha)$ be a probabilistic automaton. A symmetric
  relation $R \subseteq \dist(S)\times \dist(S)$ is a $\jan$-bisimulation if $\mu{R}\nu$ implies that
  \begin{enumerate}
  \item $\mu(\actSet,A)=\nu(\actSet,A)$ for each $\actSet\subseteq
    Act$ and $A\subseteq\AP$,
  \item for each $\actSet\subseteq Act$, whenever $\janTran{\mu}{\actSet}{\mu'}$, 
    there exists a transition $\janTran{\nu}{\actSet}{\nu'}$ such that $\mu'{R}\nu'$.
  \end{enumerate}
We write $\mu{\;\sim_{\jan}^\A\;}\nu$ if there is a 
  $\jan$-bisimulation $R$ such that $\mu{R}\nu$.
\end{definition}

For simplicity, we omit the superscript $\A$ of all relations if it is
clear from the context. Similar to Theorem~\ref{thm:ld}, we can prove that both
$\sim_{\late}$ and
$\sim_{\jan}$ are linear. In the following, we show that in
general $\sim$ is strictly coarser than $\sim_{\late}$ and
$\sim_{\jan}$, while $\sim_{\late}$ and $\sim_{\jan}$
are incomparable. Furthermore, if restricted to input-enabled
probabilistic automata, $\sim$ and $\sim_\late$ coincide.  
\begin{figure}[!t]
  \centering
  \scalebox{0.8}{
\begin{tikzpicture}[->,>=stealth,auto,node distance=2cm,semithick,scale=1,every node/.style={scale=1}]
	\tikzstyle{state}=[minimum size=0pt,circle,draw,thick]
        \tikzstyle{triangleState}=[minimum size=0pt,regular polygon, regular polygon sides=4,draw,thick]
	\tikzstyle{stateNframe}=[]
	every label/.style=draw
        \tikzstyle{blackdot}=[circle,fill=black, minimum
        size=6pt,inner sep=0pt]
     \node[blackdot,label={[label
       distance=0pt]180:{$\mu$}}](mu){$\mu$};
     \node[stateNframe](s3)[right of=mu]{};
     \node[state](s1)[above of=s3]{$s_1$};
     \node[state](s2)[below of=s3]{$s_2$};
     \node[state](s5)[right of=s3,yshift=1cm]{$s_3$};
     \node[triangleState](s6)[right of=s3,yshift=-1cm]{$s_4$};
     \path (mu) edge[dashed]     node[left] {$\frac{1}{2}$} (s1)
                edge[dashed]     node[left] {$\frac{1}{2}$} (s2)
	   (s1) edge             node[right] {$a$}   (s5)
                edge             node[left] {$a$}   (s6)
           (s2) edge             node[left] {$a$}   (s5)
                edge             node[right,yshift=-4pt]  {$b$}    (s6);
     \node[state](t1)[right of=s1,xshift=2cm]{$t_1$};
     \node[stateNframe](t2)[below of=t1] {};
     \node[state](t3)[below of=t2]{$t_2$};
       \node[blackdot,label={[label
       distance=0pt]0:{$\nu$}}](nu)[right of=t2]{$\nu$};     
     \path (nu) edge[dashed]     node[right] {$\frac{1}{2}$} (t1)
                edge[dashed]     node[right] {$\frac{1}{2}$} (t3)
    	   (t1) edge             node[left] {$a$}   (s5)
	   (t3) edge             node[left,yshift=-5pt,xshift=5pt] {$a,b$}   (s6)
                 edge             node[right] {$a$}   (s5);
\end{tikzpicture}
}  
\caption{\label{fig:sim coarser} $\mu{\;\sim\;}\nu$ but
    $\mu{\nlateBisimulation}\nu$ and $\mu{\njanBisimulation}\nu$.}
\end{figure}

\begin{figure}[!t]
  \centering
  \scalebox{0.8}{
\begin{tikzpicture}[->,>=stealth,auto,node distance=2cm,semithick,scale=1,every node/.style={scale=1}]
	\tikzstyle{state}=[minimum size=0pt,circle,draw,thick]
        \tikzstyle{triangleState}=[minimum size=0pt,regular polygon, regular polygon sides=4,draw,thick]
	\tikzstyle{stateNframe}=[]
	every label/.style=draw
        \tikzstyle{blackdot}=[circle,fill=black, minimum
        size=6pt,inner sep=0pt]
     \node[blackdot,label={[label
       distance=0pt]180:{$\mu$}}](mu){$\mu$};
     \node[state](s2)[right of=mu]{$s_2$};
     \node[state](s1)[above of=s2]{$s_1$};
     \node[state](s3)[below of=s2]{$s_3$};
     \node[state](s5)[right of=s2,yshift=1cm]{$s_5$};
     \node[state](s6)[right of=s2,yshift=-1cm]{$s_6$};
     \path (mu) edge[dashed]     node[left] {$\frac{1}{3}$} (s1)
                edge[dashed]     node {$\frac{1}{3}$} (s2)
                edge[dashed]     node[left]    {$\frac{1}{3}$}     (s3)
	   (s1) edge             node[right] {$a$}   (s5)
	   (s2) edge             node[left] {$b$}   (s5)
                edge             node[left] {$a$}   (s6)
           (s3) edge             node[left]  {$b$}        (s6);
       \node[state](t1)[right of=s3,xshift=2cm]{$t_1$};
     \node[state](t2)[above of=t1] {$t_2$};
     \node[state](t3)[above of=t2]{$t_3$};
       \node[blackdot,label={[label
       distance=0pt]0:{$\nu$}}](nu)[right of=t2]{$\nu$};     
     \path (nu) edge[dashed]     node[right] {$\frac{1}{3}$} (t1)
                edge[dashed]     node[] {$\frac{1}{3}$} (t2)
                edge[dashed]     node[right]        {$\frac{1}{3}$} (t3)
	   (t1) edge             node[left] {$a$}   (s6)
	   (t2) edge             node[left] {$b$}   (s6)
                 edge             node[left] {$a$}   (s5)
           (t3)  edge            node[left]         {$b$}  (s5)
           (s5) edge[loop below] node{$c$} (s5)
           (s6) edge[loop below] node{$d$} (s6);
\end{tikzpicture}
}
  \caption{\label{fig:jan-late} $\mu{\janBisimulation}\nu$ but $\mu{\nlateBisimulation}\nu$.}
\end{figure}

\begin{theorem}\label{thm:coincide}
  Let $\A=(S, Act, \rightarrow,L,\alpha)$ be a probabilistic
  automaton. 
  \begin{enumerate}
  \item $\sim_{\jan},\lateBisimulation{\subsetneq}\sim$.
  \item $\sim_\jan$ and $\sim_\late$ are incomparable.
  \item If $\A$ is input-enabled, $\sim_{\jan}{\subsetneq}\lateBisimulation{\equiv}\sim$.
  \end{enumerate}
\end{theorem}
\begin{proof}
  We prove the theorem in several steps:
  \begin{enumerate}
    \item $\sim_{\late}{\subseteq}\sim$. Let $$R=\{(p\cdot\mu_1 +
    (1-p)\cdot\dirac{\bot},p\cdot\nu_1 + (1-p)\cdot\dirac{\bot})\mid
    p\in[0,1]\wedge \mu_1{\;\sim_{\late}^\A\;}\nu_1\}.$$ It suffices to show
    that $R$ is a bisimulation in $\A_\bot$. Let $\mu{R}\nu$ such that
    $\mu\equiv(p\cdot\mu_1+(1-p)\cdot\dirac{\bot})$ and
    $\nu\equiv(p\cdot\nu_1+(1-p)\cdot\dirac{\bot})$ with
    $\mu_1{\;\sim_{\late}^\A\;}\nu_1$. It is easy to show that
    $\mu(A)=\nu(A)$ for any $A\subseteq\AP$. Suppose 
    $\TRANA{\mu}{a}_\bot{\mu'}$ for some $a\in Act$, we shall show that
    there exists $\TRANA{\nu}{a}_\bot{\nu'}$ such that
    $\mu'{R}\nu'$. This is trivial if $a\not\in \enableAct(s)$ for all
    $s\in \supp(\mu_1)$, or $a\in \enableAct(\mu_1)$ and
    $\consistDist{\mu_1}$. Now let 
    $\mu_1\equiv\sum_{i\in I}p_i\cdot\mu_i$ with $\mu_i$ being
    consistent for each $i\in I$. Let $J=\{i\in I \mid
    a\in\enableAct(\mu_i)\}$. By
    Definition~\ref{def:inp-enb-ext},
    for each $i\in J$ there exists $\mu'_i$ such that 
     $\TRANA{\mu_i}{a}{\mu'_i}$ and 
     $$\mu'\equiv \sum_{i\in
      J}p_i\cdot\mu'_i+(1-\sum_{i\in J}p_i)\cdot\dirac{\bot}.$$ Since
    $\mu_1{\;\sim_{\late}^\A\;}\nu_1$, there exists
    $\nu_1\equiv\sum_{i\in I}p_i\cdot\nu_i$ such that
    $\consistDist{\nu_i}$ and $\mu_i{\;\sim_{\late}^\A\;}\nu_i$ for each
    $i\in I$. Moreover, for every $i\in J$, there exists $\TRANA{\nu_i}{a}{\nu'_i}$ such
    that $\mu'_i{\;\sim_{\late}^\A\;}\nu'_i$. Now let 
    $$\nu'\equiv \sum_{i\in
      J}p_i\cdot\nu'_i+(1-\sum_{i\in J}p_i)\cdot\dirac{\bot}.$$
    Then $\TRANA{\nu}{a}_\bot{\nu'}$. From the linearity of
    $\;\sim_{\late}^\A\;$ and the
    definition of $R$, we can easily show that $\mu'{R}\nu'$,
    thus $R$ is a bisimulation in $\A_\bot$. 
  \item $\sim_{\jan}{\subseteq}\sim$. The proof is similar as the
    above case. Let $$R=\{(p\cdot\mu_1 +
    (1-p)\cdot\dirac{\bot},p\cdot\nu_1 + (1-p)\cdot\dirac{\bot})\mid
    p\in[0,1]\wedge \mu_1{\;\sim_{\jan}^\A\;}\nu_1\}.$$ Then we show
    that $R$ is a bisimulation in $\A_\bot$. Let $\mu{R}\nu$, i.e.,
    there exists $p\in [0,1]$, $\mu_1$, and $\nu_1$ such that
    $\mu\equiv p\cdot\mu_1 + (1-p)\cdot\dirac{\bot}$, $\nu\equiv
    p\cdot\nu_1 + (1-p)\cdot\dirac{\bot}$, and
    $\mu_1{\;\sim_{\jan}^\A\;}\nu_1$. Apparently, $\mu(A)=\nu(A)$ for
    any $A\subseteq AP$. Let
    $\TRANA{\mu}{a}_\bot{\mu'}$. We shall show there exists
    $\TRANA{\nu}{a}_\bot{\nu'}$ such that $\mu'{R}\nu'$. Note
    $\TRANA{\mu}{a}_\bot{\mu'}$ indicates that
    $\janTran{\mu_1}{\{a\}}{\mu'_1}$ with $\mu'\equiv
    p\cdot\mu'_1+(1-p)\cdot\dirac{\bot}$. Since
    $\mu_1{\;\sim_{\jan}^\A\;}\nu_1$, there exists
    $\janTran{\nu_1}{\{a\}}{\nu'_1}$ such that
    $\mu'_1{\;\sim_{\jan}^\A\;}\nu'_1$. Therefore there exists $\TRANA{\nu}{a}_\bot{\nu'}$
    with $\nu'\equiv p\cdot\nu'_1 + (1-p)\cdot\dirac{\bot}$. By the
    definition of $R$, $\mu'{R}\nu'$ as desired.
  \item ${\sim}\not\subseteq\sim_{\late}$ and ${\sim}\not\subseteq\sim_{\jan}$. Let $\mu$ and $\nu$ be two
    distributions as in Fig.~\ref{fig:sim coarser}, where each state
    is labelled by its shape. By adding extra transitions to the
    dead state, we can see $\mu\sim\nu$ by
    showing that the following relation is a bisimulation: 
    $
    \{(\mu,\nu), (\nu,\mu)\}\cup\mathit{ID},
    $
    where $\mathit{ID}$ denotes the identity relation.
    However, neither $\mu{\lateBisimulation}\nu$ nor
    $\mu{\janBisimulation}\nu$ holds. For the former, since $\mu$ is not consistent,
    we shall split it to $\dirac{s_1}$ and $\dirac{s_2}$ by
    Definition~\ref{def:late-bisimulation}, where $\dirac{s_1}$ cannot
    be simulated by $\nu$ and its successors. To see
    $\mu{\njanBisimulation}\nu$, let $\actSet=\{a,b\}$. Then $\mu$ can
    evolve into box states with probability 1, while the probability
    is at most 0.5 for $\nu$. 

  \item $\sim_{\jan}{\not\subseteq}\sim_{\late}$. Let $\mu$
    and $\nu$ be two distributions as in Fig.~\ref{fig:jan-late},
    where all states have the same label. Let
    $R=\{(\mu,\nu),(\nu,\mu)\}\cup\mathit{ID}$. By
    Definition~\ref{def:jan-bisimulation}, it is easy to see that 
   $R$ is a $\jan$-bisimulation. Therefore $\mu{\janBisimulation}\nu$. However,
   $\mu{\nlateBisimulation}\nu$. Since $s_1$, $s_2$, and $s_3$ have
   different enabled actions, thus $\mu$ shall be split into three
   dirac distributions, none of which can be simulated by any successor of
   $\nu$. 
  \item $\sim_{\late}{\not\subseteq}\sim_{\jan}$. Let $\mu$ and $\nu$ be the
    distributions in Fig.~\ref{fig:sim coarser} except that $s_1$ and
    $t_1$ have a transition with label $b$ to some state with a
    different label from all the others. We see that $\mu{\lateBisimulation}\nu$, since by
    adding transitions with label $b$ to $s_1$ and $t_1$, both $\mu$
    and $\nu$ are consistent, thus need not to be split. However,
    $\mu{\njanBisimulation}\nu$. To see this, let
    $\actSet=\{a,b\}$. Then $\mu$ can evolve into box states via a
    transition with label $\actSet$ with probability 1, which is not
    possible in $\nu$. 
  \item If $\A$ is input-enabled, then
    ${\sim}\subseteq{\sim_{\late}}$. Since in an input-enabled
    probabilistic automaton, all distributions are consistent, and
    there is no need to split, i.e., the last condition in
    Definition~\ref{def:late-bisimulation} is redundant. The
    counterexample
    given in the above case can be used to show that 
    $\sim_\late$ is strictly coarser than $\sim_\jan$
    even if restricted to input-enabled probabilistic automata. \qed
  \end{enumerate}
 \end{proof}

\subsection{Bisimulations and Trace Equivalences}
\label{sec:bisim-trace-equiv}
In this subsection, we discuss how different bisimulation relations
and trace equivalences are related in Segala's automata. A path $\sigma\in S\times(Act\times S)^*$ is an alternative
sequence of states and actions, and a trace $\trace\in Act^*$ is a
sequence of actions. Let $\pathSet^*(\A)$ denote the set of all finite paths of a
probabilistic automaton $\A$ and $\lastState{\path}$ the last state of
$\path$. Due to the non-determinism in a probabilistic
automaton, a \emph{scheduler} is often adopted in order to obtain a
fully probabilistic system. A scheduler can be seen as a function
taking a history execution as input, while choosing a transition as the
next step for the current state. Formally, 
\begin{definition}\label{def:scheduler}
 Let $\A=(S, Act, \rightarrow,L,\alpha)$ be a probabilistic automaton.
 A scheduler
 $\scheduler:\pathSet^*(\A)\mapsto\mathit{Dist}(Act\times\mathit{Dist}(S))$ of $\A$
  is a function such that $\scheduler(\path)(a,\mu)>0$ only if $(\lastState{\path},a,\mu)\in{\rightarrow}$.
\end{definition}


 Let $\A=(S, Act, \rightarrow,L,\alpha)$ be a probabilistic automaton, $\scheduler$ a scheduler,
$\trace$ a trace of $\A$, and $\mu$ a distribution over $S$. The probability 
of $\trace$ starting from $\mu$ under the guidance of
$\scheduler$, denoted $\Pr_\mu^\scheduler(\trace)$, is equal to $\sum_{s\in
  S}\mu(s)\cdot \Pr_s^\scheduler(\trace,s)$,
where $\Pr_s^\scheduler(\trace,\path)=1$ if $\trace$
is empty. If $\trace=a\trace'$,
$$
\Pr_s^\scheduler(\trace,\path)=\sum_{\TRANA{s}{a}{\mu}}\scheduler(\sigma)(a,\mu)\cdot\sum_{t\in
  S}\mu(t)\cdot \Pr_t^\scheduler(\trace',\path\concat(a,t)).
$$
Below follows the definition of trace distribution
equivalence~\cite{Segala-thesis}:
\begin{definition}\label{def:trace-dist-equiv}
 Let $\A=(S, Act, \rightarrow,L,\alpha)$ be a probabilistic automaton, and $\mu, \nu\in \dist(S)$. Then $\mu$ and $\nu$ are
  \emph{trace distribution equivalent}, written as $\mu\traceEquiv\nu$,
  if for each scheduler $\scheduler$ of $\A$, there exists another 
  scheduler $\scheduler'$ such that
  $\Pr_\mu^{\scheduler}(\trace)=\Pr_\nu^{\scheduler'}(\trace)$ for each
  $\trace\in Act^*$ and vice versa. 
\end{definition}

Due to the existence of non-deterministic choices in a probabilistic
automaton, we can have an alternative definition of trace distribution
equivalence, called \emph{a priori trace distribution equivalence}, by
switching the order of the qualifiers of schedulers and traces, which
resembles the definitions of \emph{a priori bisimulation} in~\cite{AlfaroMRS07,DBLP:journals/corr/abs-1107-1206}.

\begin{definition}\label{def:trace-dist-equiv-priori}
  Let $\A,\mu,$ and $\nu$ be as in Definition~\ref{def:trace-dist-equiv}. Then $\mu$ and $\nu$ are
  \emph{a priori trace distribution equivalent}, written as $\mu\traceEquivPriori \nu$,
  if for each scheduler $\scheduler$ of $\A$ and $\trace\in Act^*$, there exists another 
  scheduler $\scheduler'$ of $\A$ such that
  $\Pr_\mu^{\scheduler}(\trace)\ge \Pr_\nu^{\scheduler'}(\trace)$ and
  vice versa.
\end{definition}


\begin{figure}[!t]
  \centering
  \scalebox{0.8}{
  \begin{tikzpicture}[->,>=stealth,auto,node distance=1.7cm,semithick,scale=1, every node/.style={scale=1}]
	\tikzstyle{blackdot}=[circle,fill=black,minimum size=6pt,inner sep=0pt]
	\tikzstyle{state}=[minimum size=0pt,circle,draw,thick]
	\tikzstyle{stateNframe}=[minimum size=0pt]	
        \node[state](s0){$s_0$};
	\node[state](s1)[below of=s0]{$s_1$};
	\node[state](s2)[below left of=s1]{$s_2$};
        \node[state](s3)[below right of=s1]{$s_3$};
        \node[state](t0)[right of=s0,xshift=4cm]{$t_0$};
	\node[state](t1)[below left of=t0]{$t_1$};
	\node[state](t2)[below right of=t0]{$t_2$};
        \node[state](t3)[below of=t1]{$t_3$};
        \node[state](t4)[below of=t2]{$t_4$};
        \path (s0) edge node {$a$} (s1)
              (s1) edge node[left] {$b$} (s2)
                   edge node[right] {$c$} (s3)
              (t0) edge node[left] {$a$} (t1)
                   edge node[right] {$a$} (t2)
              (t1) edge node         {$b$} (t3)
              (t2) edge node         {$c$} (t4);

  \end{tikzpicture}
}
  \caption{\label{fig:trace-jan} $\dirac{s_0}\traceEquiv \dirac{t_0}$ but $\dirac{s_0}{\njanBisimulation}\dirac{t_0}$. }
\end{figure}

\begin{figure}[!t]
  \centering
  \scalebox{0.8}{
  \begin{tikzpicture}[->,>=stealth,auto,node distance=1.7cm,semithick,scale=1, every node/.style={scale=1}]
	\tikzstyle{blackdot}=[circle,fill=black,minimum size=6pt,inner sep=0pt]
	\tikzstyle{state}=[minimum size=0pt,circle,draw,thick]
	\tikzstyle{stateNframe}=[minimum size=0pt]	
        \node[blackdot,label={[label
       distance=0pt]90:{$\mu$}}](mu){$\mu$};
        \tikzstyle{triangleState}=[minimum size=0pt,regular polygon, regular polygon sides=4,draw,thick]
	\node[state](s1)[below left of=mu,xshift=-0.5cm]{$s_1$};
	\node[state](s2)[below right of=mu,xshift=0.5cm]{$s_2$};
        \node[state](s3)[below left of=s1]{$s_3$};
        \node[state](s4)[below right of=s1]{$s_4$};
        \node[state](s5)[below left of=s2]{$s_5$};
        \node[state](s6)[below right of=s2]{$s_6$};
        \path (mu) edge[dashed] node[left] {0.5} (s1)
                   edge[dashed] node[right] {0.5} (s2)
              (s1) edge node[left] {$a$} (s3)
                   edge node[right] {$b$} (s4)
              (s2) edge node[left] {$a$} (s5)
                   edge node[right] {$b$} (s6);
        \node[state,label={[label
       distance=0pt]90:{$\nu:=\dirac{t_0}$}}](t0)[right of=mu,xshift=5cm]{$t_0$};     
	\node[blackdot](t1)[below left of=t0,xshift=-0.5cm]{};
	\node[blackdot](t2)[below right of=t0,xshift=0.5cm]{};
        \node[state](t3)[below left of=t1]{$s_3$};
        \node[state](t4)[below right of=t1]{$s_5$};
        \node[state](t5)[below left of=t2]{$s_4$};
        \node[state](t6)[below right of=t2]{$s_6$};
        \path (t0) edge[-] node[left] {$a$} (t1)
                   edge[-] node[right] {$b$} (t2)
              (t1) edge[dashed] node[left] {0.5} (t3)
                   edge[dashed] node[right] {0.5} (t4)
              (t2) edge[dashed] node[left] {0.5} (t5)
                   edge[dashed] node[right] {0.5} (t6)
              (s3) edge[loop below] node {$c$} (s3)
              (s4) edge[loop below] node {$d$} (s4)
              (s5) edge[loop below] node {$c$} (s5)
              (s6) edge[loop below] node {$e$} (s6)
              (t3) edge[loop below] node {$c$} (t3)
              (t4) edge[loop below] node {$c$} (t4)
              (t5) edge[loop below] node {$d$} (t5)
              (t6) edge[loop below] node {$e$} (t6);
  \end{tikzpicture}
}
  \caption{\label{fig:trace-late} $\mu{\lateBisimulation}\nu$, but $\mu\not\traceEquiv \nu$.}
\end{figure}

In Definition~\ref{def:trace-dist-equiv}, we require
$\Pr_\mu^{\scheduler}(\trace)\ge \Pr_\nu^{\scheduler'}(\trace)$
instead of $\Pr_\mu^{\scheduler}(\trace)=
\Pr_\nu^{\scheduler'}(\trace)$, mainly to simplify the proofs in the
sequel. However, due to the existence of combined transitions, these
two definitions make no difference.
Directly from the definitions, $\traceEquiv$, $\traceEquivPriori$, and
$\sim$ coincide on reactive automata. For general probabilistic
automata, $\traceEquivPriori$ is strictly coarser than
$\traceEquiv$. Moreover, we have the following theorem:

\begin{theorem}\label{thm:tr-bis}
  \begin{enumerate}
  \item $\traceEquiv$ is incomparable to $\sim_\jan$,
  $\sim_\late$, and $\sim$.
\item $\sim{\subsetneq}\traceEquivPriori$.
  \end{enumerate}
 \end{theorem}
\begin{proof}
    By Theorem~\ref{thm:coincide}, to prove clause 1 it suffices to
    show $\traceEquiv{\not\subseteq}\sim$,
    $\sim_{\late}{\not\subseteq}\traceEquiv$, and
    $\sim_{\jan}{\not\subseteq}\traceEquiv$. 
  \begin{enumerate}
   \item $\traceEquiv{\not\subseteq}\sim$. Let $s_0$ and
    $t_0$ be two states as in Fig.~\ref{fig:trace-jan}. It is easy to
    see that $\dirac{s_0}\traceEquiv \dirac{t_0}$ but
    $\dirac{s_0}\not\sim \dirac{t_0}$.
  \item $\sim_{\jan}{\not\subseteq}\traceEquiv$. Let
    $\mu$ and $\nu$ be as in Fig.~\ref{fig:jan-late}. We have shown in
    the proof of Theorem~\ref{thm:coincide} that
    $\mu{\janBisimulation}\nu$. However,
    $\mu\not\traceEquiv \nu$. For instance, there exists a scheduler of $\mu$
    enabling us to see traces ``ac'', ``ad'', or ``bd'' each with
    probability $\frac{1}{3}$, which is not possible for $\nu$.
  \item $\sim_{\late}{\not\subseteq}\traceEquiv$.
    Let $\mu$ and $\nu$ be as in Fig.~\ref{fig:trace-late}. Let
    $R=\{(\mu,\nu),(\nu,\mu)\}\cup\mathit{ID}$. By
    Definition~\ref{def:late-bisimulation}, $R$ is a
    $\late$-bisimulation. Therefore $\mu{\lateBisimulation}\nu$.
    However $\mu\not\traceEquiv \nu$. For instance, there exists a
    scheduler such that from $\mu$ we will see ``$be$''
    with probability $\frac{1}{2}$ while never see ``$bd$'', but starting from
    $\nu$, the probabilities of ``$be$'' and
    ``$bd$'' are always the same. 
  \item $\sim{\subseteq}\traceEquivPriori$. By contraposition. Assume
    there exists $\mu$ and $\nu$ such that $\mu\sim\nu$ but
    $\mu\not\traceEquivPriori \nu$, i.e., there exists a scheduler
    $\scheduler$ and a trace $\trace\in Act^*$ such that for
    all schedulers $\scheduler'$,
    $\Pr_\mu^\scheduler(\trace)<\Pr_\nu^{\scheduler'}(\trace)$. Let  
    $\trace=a_0a_1\ldots a_n$. Let $\mu_0$ and $\nu_0$ be the
    corresponding distributions of $\mu$ and $\nu$ after adding extra
    transitions to the deadlock state $\bot$. Then in the input-enabled extended automaton, 
    let the chain of transitions
    $\TRANA{\mu_0}{a_0}{\TRANA{\mu_1}{a_1}{\TRANA{\ldots}{a_n}{\mu_n}}}$
    mimic $\scheduler$ as follows: For each $0\le i<n$ and
    $s\in\supp(\mu_i)$, all transitions of $s$ chosen by $\scheduler$
    with labels different from $a_i$ are switched to transitions with labels
    $a_i$. After doing so, the probability of seeing $\trace$ will not be
    lowered. By assumption, whenever $\TRANA{\nu_0}{a_0}{\TRANA{\nu_1}{a_1}{\TRANA{\ldots}{a_n}{\nu_n}}}$,
    it holds $\mu_n(\bot)<\nu_n(\bot)$, which contradicts that
    $\mu\sim\nu$. 
  
  \item $\traceEquivPriori{\not\subseteq}\sim$. The counterexample in
    Fig.~\ref{fig:trace-jan} also applies here, as
    $\dirac{s_0}\traceEquivPriori \dirac{t_0}$ but $\dirac{s_0}\not\sim \dirac{t_0}$.\qed
  \end{enumerate}
 \end{proof}


\subsection{Compositionality}
\label{sec:compositionality}
In this subsection, we discuss the compositionality of all mentioned
bisimulation relations. Let $\A_i=(S_i, Act_i, \rightarrow_i,L_i,\alpha_i)$ be two
  probabilistic automata with $i\in\{0,1\}$ and $\actSet\subseteq
  Act_1\cap Act_2$. For any $\mu_0\in\dist(S_0)$ and $\mu_1\in\dist(S_1)$, we denote by 
$\APAR{\mu_0}{\mu_1}$ a distribution over $S_0 \times S_1$, the element of which is written as
$\APAR{s_0}{s_1}$ where 
$s_0\in S_0$ and $s_1\in S_1$ for convenience, such that $(\APAR{\mu_0}{\mu_1})(\APAR{s_0}{s_1})=\mu_0(s_0)\cdot\mu_1(s_1)$. 
We recall the definition of parallel operator
of probabilistic automata given in~\cite{Segala-thesis}.
\begin{definition}\label{def:para-opr}
  Let $\A_i=(S_i, Act_i, \rightarrow_i,L_i,\alpha_i)$ be two
  probabilistic automata with $i\in\{0,1\}$ and $\actSet\subseteq
  Act_1\cap Act_2$. The parallel composition of $\A_1$ and $\A_2$ with
  respect to $\actSet$, denoted $\A_1\parallel_{\actSet} \A_2$, is a
  probabilistic automaton $(S, Act_0\cup Act_1,\rightarrow, L,
  \alpha_0\parallel_\actSet\alpha_1)$ where
  \begin{itemize}
  \item $S= S_0 \times S_1$,
  \item $s_0\parallel_{\actSet} s_1\stackrel{a}{\rightarrow}\mu_0\parallel_{\actSet}\mu_1$ with $a\in\actSet$ if
    $\forall i\in \{0,1\}$, $s_i\stackrel{a}{\rightarrow}_i \mu_i$,
  \item $s_0\parallel_{\actSet}s_1\stackrel{a}{\rightarrow}\mu_0\parallel_{\actSet}\mu_1$ with
    $a\not\in\actSet$ if $\exists i\in \{0,1\}$, $s_i\stackrel{a}{\rightarrow}_i \mu_i$ and
    $\mu_{1-i}=\dirac{s_{1-i}}$,
  \item $L(\APAR{s_0}{s_1})=L_0(s_0)\cup L_1(s_1)$.
  \end{itemize}
  \end{definition}


We say $\sim$ is compositional if for any probabilistic
automata $\A_0$ and $\A_1$, $\actSet\subseteq
Act_0\cap Act_1$, and any distributions 
$\mu_0,\mu'_0\in\dist(S_0)$ and $\mu_1\in\dist(S_1)$, whenever $\mu_0\sim\mu'_0$ in $\A_0$, we have
$\APAR{\mu_0}{\mu_1}\sim\APAR{\mu'_0}{\mu_1}$ in $\APAR{\A_0}{\A_1}$.
Similarly, we can define the notion of
compositionality for the other relations. In the following, we show that all the
three bisimulation relations mentioned in this section are unfortunately not
compositional in general:
\begin{figure}[!t]
  \centering
  \scalebox{0.8}{
  \begin{tikzpicture}[->,>=stealth,auto,node distance=1.7cm,semithick,scale=1, every node/.style={scale=1}]
	\tikzstyle{blackdot}=[circle,fill=black,minimum size=6pt,inner sep=0pt]
	\tikzstyle{state}=[minimum size=0pt,circle,draw,thick]
        \tikzstyle{stateInvisible}=[minimum size=0pt,circle,thick]
	\tikzstyle{stateNframe}=[minimum size=0pt]	
	\node[state](s31){$s_3$};
	\node[state](s41)[right of=s31]{$s_4$};
	\node[state](s32)[right of=s41]{$s_3$};
	\node[state](s42)[right of=s32]{$s_4$};
	\node[stateInvisible](r5)[right of=s42]{};
	\node[stateInvisible](r6)[right of=r5]{};
	\node[state](s11)[above of=s31]{$s_1$};
	\node[state](s21)[above of=s41]{$s_2$};
	\node[state](s12)[above of=s32]{$s_1$};
	\node[state](s22)[above of=s42]{$s_2$};
	\node[state](s5)[above of=s12]{$s_5$};
	\node[state](s6)[above of=s22]{$s_6$};
	\node[state](r3)[above of=r5]{$r_3$};	
	\node[state](r4)[above of=r6]{$r_4$};
	\node[state](r1)[above of=r3]{$r_1$};	
	\node[state](r2)[above of=r4]{$r_2$};
	\node[stateNframe](r00)[above of=r1]{};
	\node[state](r0)[right of=r00,xshift=-0.9cm]{$r_0$};
	\node[blackdot](d1)[above of=s11,xshift=0.9cm]{};  
	\node[state](s0)[above of=d1]{$s_0$};
	\node[stateNframe,scale=0.9](s0')[above of=s5]{};
	\node[blackdot,label={[label
       distance=0pt]90:{$\mu$}}](d2)[right of=s0',xshift=-0.9cm]{};
	\path (s0) edge[-]							node {$a$} (d1)
			  (d1) edge[dashed]				node[left] {$\frac{1}{2}$} (s11)
			       edge[dashed] 				node[yshift=-9pt] {$\frac{1}{2}$} (s21)
			  (d2) edge[dashed]				node[left] {$\frac{1}{2}$} (s5)
			       edge[dashed]				node[yshift=-9pt] {$\frac{1}{2}$} (s6)
			  (s5) edge								node[left] {$a$} (s12)
			  (s6) edge								node {$a$} (s22)
			  (s11) edge								node[left] {$b$} (s31)
			  (s21) edge								node {$c$} (s41)
			  (s12) edge								node[left] {$b$} (s32)
			  (s22) edge								node {$c$} (s42)
			  (r0) edge								node[left] {$a$} (r1)
                               edge								node[yshift=-8pt] {$a$} (r2)
			  (r1) edge								node[left] {$b$} (r3)
			  (r2) edge								node {$c$} (r4);
  \end{tikzpicture}
}
  \caption{\label{fig:non-comp} $\dirac{s_0}\sim \mu$, but $(\dirac{s_0}\parallel_{\actSet}\dirac{r_0})\not\sim (\mu\parallel_{\actSet}\dirac{r_0})$, where $\actSet=\{a,b,c\}$.}
\end{figure}

\begin{theorem}\label{thm:non-comp}
  $\sim$, $\sim_\late$, and $\sim_\jan$ are not compositional.
\end{theorem}
\begin{proof}
  We only show the non-compositionality of $\sim$, since the other two
   can be proved in a similar way. Let $s_0$, $\mu$, and
  $r_0$ be as in Fig.~\ref{fig:non-comp}. It is easy to see that
  $\dirac{s_0}\sim\mu$. However, when composing $\dirac{s_0}$ and
  $\mu$ with $\dirac{r_0}$ by enforcing synchronization on
  $\actSet=\{a,b,c\}$, we have
  $(\dirac{s_0}\parallel_{\actSet}\dirac{r_0})\not\sim (\mu\parallel_{\actSet}\dirac{r_0})$. For
  instance, $\APAR{\mu}{\dirac{r_0}}$ can reach the distribution
  $\frac{1}{2}\dirac{\APAR{s_1}{r_1}}+\frac{1}{2}\dirac{\APAR{s_2}{r_2}}$,
  from which transitions with label $b$ or $c$ 
  are both available, each with probability $\frac{1}{2}$. However, this is not possible in
  $\dirac{s_0}\parallel_{\actSet}\dirac{r_0}$, where transitions with label $b$ or $c$ cannot be enabled at the
  same time. 
\qed \end{proof}

Although the three bisimulations are not
compositional in general, we show that, by restricting to a subclass
of schedulers, they are compositional. For this, we need to introduce some
notations. Let $\distTran{}{}{}$ be a transition relation on distributions
defined as follows: In case $\mu$ is sequential, i.e., none of the states in its support
has a Cartesian form, $\distTran{\mu}{a}{\mu'}$ iff
$\TRANA{\mu}{a}{\mu'}$ as in Definition~\ref{def:dptran}.
   Otherwise if $\mu=\APAR{\mu_0}{\mu_1}$ for
some $\actSet\subseteq Act$, then
 $\distTran{\mu}{a}{\mu'}$ iff 
 \begin{itemize}
 \item either $a\in\actSet$, and for all $i\in\{0,1\}$, there exists
 $\distTran{\mu_i}{a}{\mu'_i}$ such that $\mu'=\APAR{\mu'_0}{\mu'_1}$, 
\item or $a\not\in\actSet$, and there exists $i\in\{0,1\}$ and
  $\distTran{\mu_i}{a}{\mu'_i}$ such that $\mu'=\APAR{\mu'_0}{\mu'_1}$
  with $\mu'_{1-i}=\mu_{1-i}$.
 \end{itemize} 
Intuitively, $\distTran{}{}{}$ is subsumed by $\TRANA{}{}{}$ such that
transitions of a distribution $\APAR{\mu_0}{\mu_1}$ can be projected to transitions of
$\mu_0$ and $\mu_1$. The definition of $\distTran{}{}{}$ can be generalized to
distributions composed of more than 2 distributions. 
It is worthwhile to mention that the definition of $\distTran{}{}{}$
is not ad hoc. Actually, it coincides with transitions induced by
\emph{distributed schedulers}~\cite{DBLP:conf/formats/GiroD07}.   
It has been argued by many authors, see
e.g.~\cite{Segala-thesis,de1999verification,DBLP:conf/formats/GiroD07},
that schedulers defined in 
Definition~\ref{def:scheduler} are too powerful in certain
scenarios. For this, a subclass of schedulers, called \emph{distributed schedulers},
was introduced in~\cite{DBLP:conf/formats/GiroD07} to restrict the
power of general schedulers. Instead of giving
the formal definition of distributed schedulers, we illustrate the underlying
idea by an example. We refer interested readers to~\cite{DBLP:conf/formats/GiroD07} for more
details.
\begin{example}\label{ex:non-comp}
  Let $s_0,r_0,\mu$ be as in Fig.~\ref{fig:non-comp}. We have shown in
  Theorem~\ref{thm:non-comp} that
  $\APAR{\dirac{s_0}}{\dirac{r_0}}\not\sim \APAR{\mu}{\dirac{r_0}}$
  with $\actSet=\{a,b,c\}$. The main reason is that in
  $\APAR{\mu}{\dirac{r_0}}$, there exists a scheduler such that it chooses
  the left $a$-transition of $r_0$ when at $\APAR{s_5}{r_0}$, while it
  chooses the right $a$-transition of $r_0$ when at
  $\APAR{s_6}{r_0}$. This scheduler cannot be simulated by any
  scheduler of $\APAR{\dirac{s_0}}{\dirac{r_0}}$. However, such a
  scheduler is not distributed, since it corresponds to two different
  transitions of $r_0$, thus cannot ``distribute'' its choices to
  states $s_0$ and $r_0$.
\end{example}

 Let $\distributedSchedulers$ denote the set of all distributed
 schedulers. It is easy to check that distributed
 schedulers induce exactly transitions in $\distTran{}{}{}$.
Even though $\sim$ is not compositional in general, we show that
it is compositional if restricted to distributed schedulers, similarly for
$\sim_\late$ and $\sim_\jan$.
Below we redefine the bisimulation relations with restricted
to schedulers in $\distributedSchedulers$, which is almost the same as
Definition~\ref{def:bisi} except that all transitions under
consideration must be induced by a distributed scheduler.
\begin{definition}\label{def:bisi-dist-sch}
  Let $\A=(S, Act, \rightarrow,L,\alpha)$ be an input-enabled probabilistic automaton.  A symmetric
  relation $R \subseteq \dist(S)\times \dist(S)$ is a
  (distribution-based) bisimulation with respect to $\distributedSchedulers$ if $\mu{R}\nu$ implies that
  \begin{enumerate}
  \item $\mu(A)=\nu(A)$ for each $A\subseteq AP$, and
  \item for each $a\in Act$, whenever $\distTran{\mu}{a}{\mu'}$
    then there exists $\distTran{\nu}{a}{\nu'}$ such that $\mu' R \nu'$.
  \end{enumerate}
We write $\mu{\;\sim^\A_\distributedSchedulers\;}\nu$ if there is a 
  bisimulation $R$ with respect to $\distributedSchedulers$ such that $\mu R \nu$.
\end{definition}

In an analogous way, we can also define the restricted version of
$\sim_\late$, $\sim_\jan$, and $\sim_{\epsilon}$, denoted
$\sim_{(\late,\distributedSchedulers)}$,
$\sim_{(\jan,\distributedSchedulers)}$, and $\sim_{(\epsilon,\distributedSchedulers)}$ respectively.
Below we show
that by restricting to distributed schedulers, $\sim$,
$\sim_\late$, and $\sim_\jan$ are all compositional.

\begin{theorem}\label{thm:bisi-comp}
  $\sim_{\distributedSchedulers}$, $\sim_{(\late,\distributedSchedulers)}$, and
$\sim_{(\jan,\distributedSchedulers)}$ are compositional. 
\end{theorem}
\begin{proof}
  We only prove the compositionality of $\sim_\distributedSchedulers$
  here, as the proofs for the other cases are similar. 
  Let $\A_i=(S_i, Act_i, \rightarrow_i,L_i,\alpha_i)$ be two
  probabilistic automata with $i\in\{0,1\}$ and $\actSet\subseteq
  Act_1\cap Act_2$. Let
  $R=\{(\mu_0\parallel_\actSet\mu_1,\nu_0\parallel_\actSet\mu_1)\mid
  \mu_0\sim_\distributedSchedulers\nu_0\}$, where
  $\mu_0,\nu_0\in\dist(S_0)$ and $\mu_1\in\dist(S_1)$. It suffices to show that 
  $R$ is a bisimulation with respect to $\distributedSchedulers$.
  Let $(\APAR{\mu_0}{\mu_1}){R}(\APAR{\nu_0}{\mu_1})$. Obviously,
  $(\APAR{\mu_0}{\mu_1})(A) = (\APAR{\nu_0}{\mu_1})(A)$  for each $A\subseteq AP$, since, say,
  $(\APAR{\mu_0}{\mu_1})(A)=\sum_{B,B'.B\cup
    B'=A}\mu_0(B)\cdot\mu_1(B')$. Let
  $\distTran{\APAR{\mu_0}{\mu_1}}{a}{\mu'}$. We show that there
  exists $\distTran{\APAR{\nu_0}{\mu_1}}{a}{\nu'}$ such that $\mu'{R}\nu'$. We
  distinguish two cases: 
  \begin{enumerate}
  \item $a\not\in\actSet$:
    According to the definition of $\distTran{}{}{}$, either
    \begin{inparaenum}[\scriptsize$(i)$]
      \item $\distTran{\mu_0}{a}{\mu'_0}$ such that
        $\mu'\equiv\APAR{\mu'_0}{\mu_1}$, or
      \item $\distTran{\mu_1}{a}{\mu'_1}$ such that $\mu'\equiv\APAR{\mu_0}{\mu'_1}$, 
    \end{inparaenum}
    We first consider case {\scriptsize$(i)$}. Since
    ${\mu_0}\sim_{\distributedSchedulers}{\nu_0}$, 
    there exists $\distTran{\nu_0}{a}{\nu'_0}$ such that $\mu'_0\sim_{\distributedSchedulers}\nu'_0$.
    Therefore $\distTran{\APAR{\nu_0}{\mu_1}}{a}{\APAR{\nu'_0}{\mu_1}}$.
    According to the definition of $R$, we have $\mu'\equiv(\APAR{\mu'_0}{\mu_1}){R}(\APAR{\nu'_0}{\mu_1})\equiv\nu'$ as desired.
    The proof of case {\scriptsize$(ii)$} is similar and omitted here.
  \item $a\in\actSet$:
    It must be the case that 
    $\distTran{\mu_0}{a}{\mu'_0}$ and
    $\distTran{\mu_1}{a}{\mu'_1}$ such that
    $\mu'\equiv\APAR{\mu'_0}{\mu'_1}$. 
    Since $\mu_0\sim_{\distributedSchedulers}\nu_0$,
    there exists $\distTran{\nu_0}{a}{\nu'_0}$ such that
    $\mu'_0\sim_{\distributedSchedulers}\nu'_0$. Hence there exists $\distTran{\APAR{\nu_0}{\mu_1}}{a}{\APAR{\nu'_0}{\mu'_1}}$
    such that $\mu'\equiv(\APAR{\mu'_0}{\mu'_1}){R}(\APAR{\nu'_0}{\mu'_1})\equiv\nu'.$ \qed
  \end{enumerate}
 \end{proof}

\begin{example}\label{ex:comp-dist}
  Let $s_0,r_0,$ and $\mu$ be as in Example~\ref{ex:non-comp}. Since
  $s_0$ and $\mu$ are sequential, $\sim_{\distributedSchedulers}$
  degenerates to $\sim$, and
  $\dirac{s_0}\sim_{\distributedSchedulers}\mu$. We can also show
  that
  $\APAR{\dirac{s_0}}{\dirac{r_0}}\sim_{\distributedSchedulers}\APAR{\mu}{\dirac{r_0}}$
  with $\actSet=\{a,b,c\}$. Intuitively,
  by restricting to distributed schedulers, $r_0$ cannot choose different
  transitions when at $\APAR{s_5}{r_0}$ or $\APAR{s_6}{r_0}$, thus
  transitions with label $b$ or $c$ cannot be enabled at the same time. 
\end{example}

Since the bisimilarity $\sim$ can be seen as a special case of approximate
bisimulation with $\epsilon=0$, approximate bisimulation is in
general not compositional either. However, by restricting to
distributed schedulers, the compositionality also holds for (discounted) approximate
bisimulations.

\begin{theorem}\label{thm:approx-bisi-comp}
  $\sim_{(\epsilon,\distributedSchedulers)}$ is compositional for any $\gamma\in(0,1]$.
\end{theorem}
\begin{proof}
  Let $\A_i=(S_i, Act_i, \rightarrow_i,L_i,\alpha_i)$ be two
  probabilistic automata with $i\in\{0,1\}$ and $\actSet\subseteq
  Act_1\cap Act_2$. 
  Let $\{R_{(\epsilon,\distributedSchedulers)} \mid \epsilon \geq 0\}$, where 
  $$R_{(\epsilon,\distributedSchedulers)} =\{(\APAR{\mu_0}{\mu_1},\APAR{\nu_0}{\mu_1})\mid
  {\mu_0}~{\sim_{(\epsilon,\distributedSchedulers)}}~{\nu_0}\},$$ be a
  family of relations on $\dist(S_0)\times \dist(S_1)$. It suffices
  to show that each $R_{(\epsilon,\distributedSchedulers)}$ is a (discounted) approximate
  bisimulation with respect to $\distributedSchedulers$. Note that
  $d_{AP}(\APAR{\mu_0}{\mu_1},\APAR{\nu_0}{\mu_1})=d_{AP}(\mu_0,\nu_0)$. 
  The remaining proof is analogous to the proof of
  Theorem~\ref{thm:bisi-comp} and omitted here.
\qed \end{proof}

A direct consequence of Theorem~\ref{thm:approx-bisi-comp} is that the
bisimulation distance is non-expansive under parallel operators, if
restricted to distributed schedulers, i.e.,
$D_b(\APAR{\mu_0}{\mu_1},\APAR{\nu_0}{\mu_1})\le D_b(\mu_0,\nu_0)$ for
any $\mu_0,\nu_0,$ and $\mu_1$.

\section{Decidability and Complexity}
\label{sec:decision-algorithms}
It has been proved in~\cite[Lem. 1]{HermannsKK14} that every linear
bisimulation $R$ corresponds to a bisimulation matrix $E$ of size $n\times
m$ with $n=\ABS{S}$ and $1\le m \le n$. Two distributions $\mu$ and
$\nu$ are related by $R$ iff $(\mu-\nu)E=0$, where distributions are
seen as vectors. Furthermore, by making use of the linear structure, a decision algorithm was
presented in~\cite{HermannsKK14} for $\sim_\jan$. It was also
mentioned that, with slight changes, this algorithm can be applied to
deal with both $\sim$ and $\sim_\late$. Interested readers can refer
to~\cite{HermannsKK14} for details about the algorithm. However,
we show in this section that the problem of deciding
\emph{approximate} bisimulation is more difficult: it is in fact
undecidable when no discounting is permitted, while the discounted
version is decidable but NP-hard.

In the remaining part of this section, we shall focus on approximate
bisimulations with and without discounting. We first recall the following
undecidable problem~\cite{Paz:1971:IPA:1097027}. 
\begin{theorem}\label{thm:undecidable-pfa}
  Let $\A$ be a reactive automaton and $\epsilon\in(0,1)$. The
  following problem is undecidable: Whether $\A$ is  $\epsilon$-empty, i.e., whether there exists $w\in Act^*$
  such that $\A(w)>\epsilon$.
\end{theorem}

By making use of the following reduction, we show that approximate bisimulation
without discounting is undecidable.
\begin{lemma}\label{lem:reduction}
  Let $\A$ be a reactive automaton with initial distribution $\alpha$
  and $\epsilon\in(0,1)$. Let $s$ be a state such that
  $L(s)=\emptyset$ and $\TRANA{s}{a}{\dirac{s}}$ for all $a\in
  Act$. Then $\A$ is $\epsilon$-empty iff $\alpha\sim_\epsilon\dirac{s}$.
\end{lemma}
\begin{proof}
  \begin{enumerate}
  \item Suppose $\A$ is $\epsilon$-empty. We show that
    $\alpha\sim_\epsilon\dirac{s}$. Assume
    $\alpha\not\sim_\epsilon\dirac{s}$. By construction, it must be
    the case that from $\alpha$ a distribution $\mu$ is reached in
    finite steps such that $\mu(F)>\epsilon$ with $F\subseteq S$ being
    the set of accepting states, which contradicts that $\A$ is $\epsilon$-empty. 
  \item Suppose $\alpha\sim_\epsilon\dirac{s}$. We show that $\A$
    must be $\epsilon$-empty. By contraposition, suppose 
    there exists $w\in Act^*$ such that $\A(w)>\epsilon$. Let $w=a_0a_1\ldots
    a_n$. This means that there exits
    $\TRANA{\alpha}{a_0}{\TRANA{\mu_0}{a_1}{\TRANA{\ldots}{a_n}{\mu_n}}}$
    such that $\mu_n(F)>\epsilon$. Since $\dirac{s}$ can only reach
    itself, we have $\alpha\not\sim_\epsilon\dirac{s}$, a contradiction.\qed
  \end{enumerate}
 \end{proof}

Directly from Theorem~\ref{thm:undecidable-pfa} and
Lemma~\ref{lem:reduction}, we reach a proposition as below showing the
undecidability of approximate bisimulation without discounting.
\begin{proposition}\label{prop:decidability}
  The following problem is undecidable:
   Given
  $\epsilon\in(0,1)$ and $\mu,\nu\in\mathit{Dist}(S)$,
  decide whether $\mu\sim_\epsilon\nu$ without discounting.  
\end{proposition}

For discounted approximate bisimilarity, the problem turns out to be
decidable. Instead of presenting the algorithm formally in this paper, we only
sketch how the algorithm works. Intuitively, in the definition of
discounted approximate bisimulation, the distance $\epsilon$ is
discounted with $\gamma$ at each step. Since $\gamma$ is strictly less
than 1, $\epsilon$ will for sure become larger than or equal to 1 in finite steps, in which
case $\mu\sim_\epsilon\nu$ for any $\mu$ and
$\nu$. This enables us to identify a finite set of pivotal
distributions, which contains enough information for deciding
discounted approximate bisimulation in a probabilistic automaton.
However, we also note that the algorithms presented
in~\cite{ChenBW12,Fu12} for computing \emph{state-based} approximate
bisimilarity cannot be applied here. Even though we can identify a
finite set of pivotal distributions, for each pivotal distribution
there are infinitely many distributions approximately bisimilar with
it. Therefore in the algorithm these infinite sets of distributions
have to be represented symbolically, which makes the whole algorithm
very involved. Actually we can show that deciding
discounted approximate bisimulation is 
NP-hard.
\begin{theorem}\label{thm:np-hard}
  The following decision problem is \emph{NP}-hard: Given
  $\epsilon,\gamma\in(0,1)$ and $\mu,\nu\in\mathit{Dist}(S)$,
  decide whether $\mu\sim_\epsilon\nu$ with discounting factor $\gamma$.
\end{theorem}
\begin{proof}
  Firstly, we recall the following NP-hard problem from, say, \cite{Garey:1990:CIG:574848}: Given an undirected graph
  $G=(V,E)$ and $k\le\ABS{V}$, decide whether there exists a clique
  in $G$ with size larger than $k$. Note a clique is a sub-graph where
  every two vertexes are connected. We shall reduce the clique
  checking problem to the problem of deciding discounted approximate
  bisimulation. Our reduction is almost the same
  as~\cite{DBLP:journals/jcss/LyngsoP02}, where the clique checking
  problem was reduced to the problem of deciding the consensus string
  in a hidden Markov chain. We sketch the construction as below:
Fix an order over vertexes in $V=\{a_1,\ldots,a_n\}$ with $n=\ABS{V}$.
Let  $\A_G =  (S, Act, \rightarrow,L,\alpha)$ be a reactive probabilistic automaton such that 
\begin{enumerate}

\item $S = \cup_{1\le i\le n} S_i\cup\{s, t, r\}$ where for each $i$, 
  $S_i = \{s^j_i\}_{1\le j\le n}$. To simplify the presentation, we let
  $s^{n+1}_i=t$ for each $i$ in the sequel;

\item $Act = V\cup\{\tau\}$ and  $\alpha = \dirac{s}$;
 
\item The transition relation $\rightarrow$ is defined as follows: 
\begin{enumerate}
\item $\TRANA{s}{\tau}{\mu}$ such that for
  $1\le i\le n$, $\mu(s^1_i)=\frac{\lambda_i}{\lambda}$, where
  $\lambda=\sum_{1\le i\le n} \lambda_i$ with
  $\lambda_i=2^{\mathit{deg}(a_i)}$ and $\mathit{deg}(a_i)$ being the degree of vertex
  $a_i$;
\item  
  for each $1\le i, j\le n$, we distinguish several
  cases: If $j=i$, then
  $\TRANA{s^j_i}{a_i}{s^{j+1}_i}$; If $j\neq i$ and $(a_i,a_j)\not\in E$, then $\TRANA{s^j_i}{\tau}{s^{j+1}_i}$; If $j\neq i$ and $(a_i,a_j)\in E$, 
  then both $\TRANA{s^j_i}{a_j}{\mu}$ and $\TRANA{s^j_i}{\tau}{\mu}$ where
  $\mu(s^{j+1}_i)=\mu(r)=\frac{1}{2}$;
 \item All other transitions not stated in the above items have the form
  $\TRANA{u}{a}{\dirac{r}}$ for $u\in S$ and $a\in Act$;
\end{enumerate}
  \item $L(t) = AP$ and $L(u) = \emptyset$ for each $u\neq t$. That is, $t$ is the \emph{only} accepting state.
\end{enumerate}
To help understanding the construction, we present in Fig.~\ref{fig:clique}
the probabilistic automaton corresponding to the undirected
graph depicted in Fig.~\ref{fig:undirect-graph}. For simplicity, 
the state $r$ and all transitions leading to it are omitted.
The order over vertexes is defined by $a\leq b\leq c\leq d$, and
the four branchings of $s$ after performing action $\tau$ correspond to, from left to right,
$d$, $c$, $b$, $a$, respectively. The example is taken 
from~\cite{DBLP:journals/jcss/LyngsoP02}.

 By construction, all paths in $\A_G$ able to reach $t$ are of
 length $n+1$ and have the same probability $\frac{1}{\lambda}$. The
 size of the maximal clique in $G$ is $k$ iff there exists
 $\TRANA{\dirac{s}}{\tau}{\TRANA{}{b_1}{\ldots\TRANA{}{b_n}{\mu}}}$
 such that $\mu(t)=\frac{k}{\lambda}$ and $|\{b_i\}_{1\le i\le
   n}\setminus\{\tau\}|=k$. Moreover, the set $\{b_i\}_{1\le i\le
   n}\setminus\{\tau\}$ constitutes the maximal clique in $G$. Note that
 $\mu$ is also the distribution reachable from $\dirac{s}$ where the
 probability of $t$ is maximal. Since such $\mu$ can only be reached from
 $\dirac{s}$ after performing $n+1$ transitions, we have
 $\dirac{r}{\sim_{\epsilon}}\dirac{s}$ for any
 $\epsilon\ge\gamma^{n+1}\cdot\frac{k}{\lambda}$. Therefore, for any given 
 $\gamma\in (0,1)$,
 the size of the maximal clique of $G$ is $k$ iff
 $\dirac{r}\sim_{(\gamma^{n+1}\cdot\frac{k}{\lambda})}\dirac{s}$ but
 $\dirac{r}\not\sim_{(\gamma^{n+1}\cdot\frac{k-1}{\lambda})}\dirac{s}$. \qed 

%
%

  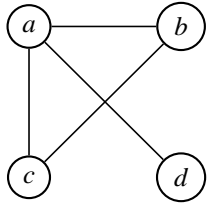
\begin{figure}[!t]
  \centering
  \scalebox{1}{
  \begin{tikzpicture}[->,>=stealth,auto,node distance=2cm,semithick,scale=1, every node/.style={scale=1}]
	\tikzstyle{blackdot}=[circle,fill=black,minimum size=6pt,inner sep=0pt]
	\tikzstyle{state}=[minimum size=0pt,circle,draw,thick]
	\tikzstyle{stateNframe}=[minimum size=0pt]
	\node[state](s1){$a$};
        \node[state](s2)[right of=s1]{$b$};
        \node[state](s3)[below of=s1]{$c$};
        \node[state](s4)[right of=s3]{$d$};
        \path (s1) edge[-] node {} (s2)
                   edge[-] node {} (s3)
                   edge[-] node {} (s4)
              (s2) edge[-] node {} (s3);
   \end{tikzpicture}
}
  \caption{\label{fig:undirect-graph} An undirected graph with
    $\{a,b,c\}$ being the maximal clique.}
\end{figure}

  \begin{figure}[!t]
    \def\dummyAct{\tau}
    \def\hlc{red}
  \centering
  \scalebox{0.8}{
  \begin{tikzpicture}
    \tikzstyle{state}=[circle,black,draw,minimum size=8pt,inner
    sep=0pt]
    \tikzstyle{blackdot}=[circle,fill=black,minimum size=2pt,inner
    sep=0pt]
    \tikzstyle{blackbox}=[rectangle,draw,minimum size=8pt,inner sep=0pt]
    \matrix [matrix of math nodes,nodes in empty cells, row sep=4em,
    column sep=2.5em, ampersand replacement=\&] (m){
    \& \& \& \& \& \& \node[state](m0-7){s}; \& \& \& \& \& \& \\
    \& \& \& \& \& \& \node[blackdot](m1-7){}; \& \& \& \& \& \& \\
    \& \node[state](m2-2){};\& \& \& \node[state](m2-5){}; \& \& \& \& \node[state](m2-9){}; \& \& \& \node[state](m2-12){};\& \\
    \node[blackdot](m3-1){};\& \&\node[blackdot](m3-3){}; \&
    \node[blackdot](m3-4){};\& \& \node[blackdot](m3-6){}; \& \&
    \node[blackdot](m3-8){}; \& \& \node[blackdot](m3-10){}; \& \& \node[state](m3-12){}; \& \\
    \& \node[state](m4-2){};\& \& \& \node[state](m4-5){};\& \& \& \&
    \node[state](m4-9){}; \& \& \node[blackdot](m4-11){}; \& \& \node[blackdot](m4-13){}; \\
    \& \node[state](m5-2){}; \& \& \node[blackdot](m5-4){}; \& \&
    \node[blackdot](m5-6){}; \& \& \& \node[state](m5-9){}; \& \& \& \node[state](m5-12){}; \& \\
    \& \& \& \& \node[state](m6-5){}; \& \& \& \node[blackdot](m6-8){}; \&
    \& \node[blackdot](m6-10){}; \& \node[blackdot](m6-11){};\& \& \node[blackdot](m6-13){};\\
    \& \node[state](m7-2){}; \& \& \& \node[state](m7-5){}; \& \& \& \&  \node[state](m7-9){};
    \& \& \& \node[state](m7-12){};\& \\
    \& \& \& \& \& \& \& \& \& \& \node[blackdot](m8-11){}; \& \& \node[blackdot](m8-13){}; \\
    \& \node[blackbox](m9-2){t};\& \& \& \node[blackbox](m9-5){t};\& \& \&
    \&  \node[blackbox](m9-9){t}; \& \& \& \node[blackbox](m9-12){t};\& \\
    };
    \path[-stealth]
    (m0-7) edge[-] node[]{\colorbox{white}{$\dummyAct$}} (m1-7)
    (m1-7) edge[dashed] node[]{\colorbox{white}{$\frac{2}{18}$}} (m2-2) edge[dashed]
    node[]{\colorbox{white}{$\frac{4}{18}$}} (m2-5) edge[dashed] node[]{\colorbox{white}{$\frac{4}{18}$}} (m2-9)
    edge[dashed] node[]{\colorbox{white}{$\frac{8}{18}$}} (m2-12)
    (m2-2) edge[-] node[]{\colorbox{white}{$a$}} (m3-1) edge[-] node[]{\colorbox{white}{$\dummyAct$}} (m3-3)
    (m2-5) edge[-] node[]{\colorbox{white}{$a$}} (m3-4) edge[-] node[]{\colorbox{white}{$\dummyAct$}} (m3-6)
    (m2-9) edge[-] node[]{\colorbox{white}{$a$}} (m3-8) edge[-] node[]{\colorbox{white}{$\dummyAct$}} (m3-10)
    (m2-12) edge node[]{\colorbox{white}{$a$}} (m3-12)
    (m3-1) edge[dashed] node[]{\colorbox{white}{$\frac{1}{2}$}} (m4-2)
    (m3-3) edge[dashed] node[]{\colorbox{white}{$\frac{1}{2}$}} (m4-2)
    (m3-4) edge[dashed] node[]{\colorbox{white}{$\frac{1}{2}$}} (m4-5)
    (m3-6) edge[dashed] node[]{\colorbox{white}{$\frac{1}{2}$}} (m4-5)
    (m3-8) edge[dashed] node[]{\colorbox{white}{$\frac{1}{2}$}} (m4-9)
    (m3-10) edge[dashed] node[]{\colorbox{white}{$\frac{1}{2}$}} (m4-9)
    (m3-12) edge[-] node[]{\colorbox{white}{$b$}} (m4-11) edge[-] node[]{\colorbox{white}{$\dummyAct$}} (m4-13)
    (m4-2) edge node[]{\colorbox{white}{$\dummyAct$}} (m5-2)
    (m4-5) edge[-] node[]{\colorbox{white}{$b$}} (m5-4) edge[-] node[]{\colorbox{white}{$\dummyAct$}} (m5-6)
    (m4-9) edge node[]{\colorbox{white}{$b$}} (m5-9)
    (m4-11) edge[dashed] node[]{\colorbox{white}{$\frac{1}{2}$}} (m5-12)
    (m4-13) edge[dashed] node[]{\colorbox{white}{$\frac{1}{2}$}} (m5-12)
    (m5-4) edge[dashed] node[]{\colorbox{white}{$\frac{1}{2}$}} (m6-5)
    (m5-6) edge[dashed] node[]{\colorbox{white}{$\frac{1}{2}$}} (m6-5)    
    (m5-9) edge[-] node[]{\colorbox{white}{$c$}} (m6-8) edge[-] node[]{\colorbox{white}{$\dummyAct$}} (m6-10)
    (m5-12) edge[-] node[]{\colorbox{white}{$c$}} (m6-11) edge[-] node[]{\colorbox{white}{$\dummyAct$}} (m6-13)
    (m6-5) edge node[]{\colorbox{white}{$c$}} (m7-5)
    (m6-8) edge[dashed] node[]{\colorbox{white}{$\frac{1}{2}$}} (m7-9)
    (m6-10) edge[dashed] node[]{\colorbox{white}{$\frac{1}{2}$}} (m7-9)
    (m6-11) edge[dashed] node[]{\colorbox{white}{$\frac{1}{2}$}} (m7-12)
    (m6-13) edge[dashed] node[]{\colorbox{white}{$\frac{1}{2}$}} (m7-12)    
    (m7-12) edge[-] node[]{\colorbox{white}{$d$}} (m8-11) edge[-] node[]{\colorbox{white}{$\dummyAct$}} (m8-13)
    (m8-11) edge[dashed] node[]{\colorbox{white}{$\frac{1}{2}$}} (m9-12)
    (m8-13) edge[dashed] node[]{\colorbox{white}{$\frac{1}{2}$}} (m9-12)
    (m5-2) edge node[]{\colorbox{white}{$\dummyAct$}} (m7-2)
    (m7-2) edge node[]{\colorbox{white}{$d$}} (m9-2)
    (m7-5) edge node[]{\colorbox{white}{$\dummyAct$}} (m9-5)
    (m7-9) edge node[]{\colorbox{white}{$\dummyAct$}} (m9-9);
   \end{tikzpicture}
}
  \caption{\label{fig:clique} A reactive PA corresponding to the
    graph in Fig.~\ref{fig:undirect-graph}.}
\end{figure}
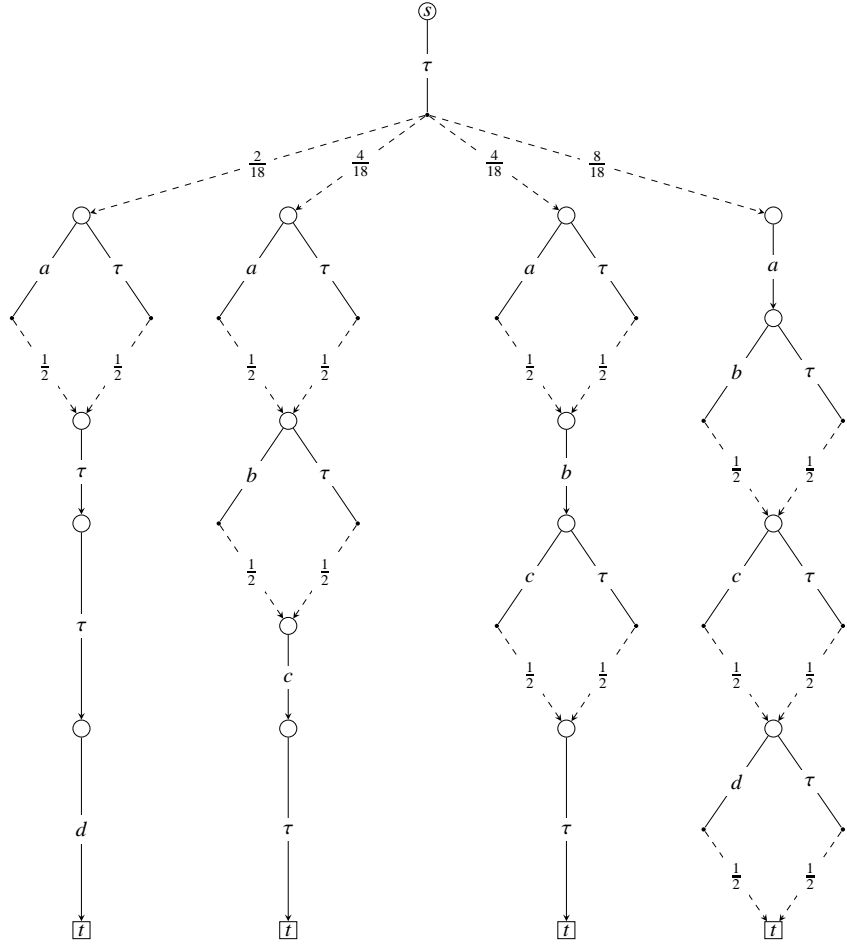
  \end{proof}

\section{Discussion and Future Work}\label{sec:conclusion}
In this paper, we considered Segala's automata, and proposed a novel
notion of bisimulation by joining the existing notions of equivalence
and bisimilarities. Our relations are defined over
distributions. We have compared our bisimulation to some
existing distribution-based bisimulations and discussed their
compositionality and relations to trace equivalences. We have
demonstrated the utility of our definition by studying
distribution-based bisimulation metrics, which have been extensively
studied for MDPs in state-based case. The decidability and complexity of deciding
approximate bisimulations with or without discounting were also
discussed.

State-based bisimulation has proven to be a powerful state space
reduction technique in model checking. As future work we would like to
study how distribution-based bisimulations can be used to accelerate
probabilistic model checking. One may combine it with state-based
bisimulation which has efficient decision procedure, or component-based verification technique.  As another direction of future work we
would like to investigate weaker preorder relations such as
simulations between distributions.

\section*{References}


\begin{thebibliography}{10}
\expandafter\ifx\csname url\endcsname\relax
  \def\url#1{\texttt{#1}}\fi
\expandafter\ifx\csname urlprefix\endcsname\relax\def\urlprefix{URL }\fi
\expandafter\ifx\csname href\endcsname\relax
  \def\href#1#2{#2} \def\path#1{#1}\fi

\bibitem{Rabin63}
M.~Rabin, {Probabilistic automata}, Information and Control 6~(3) (1963)
  230--245.

\bibitem{Tzeng92}
W.~Tzeng, {A polynomial-time algorithm for the equivalence of probabilistic
  automata}, SIAM Journal on Computing 21~(2) (1992) 216--227.

\bibitem{KieferMOWW11}
S.~Kiefer, A.~S. Murawski, J.~Ouaknine, B.~Wachter, J.~Worrell, {Language
  Equivalence for Probabilistic Automata}, in: CAV, Vol. 6806 of Lecture Notes
  in Computer Science, Springer, 2011, pp. 526--540.

\bibitem{KieferMOWW12}
S.~Kiefer, A.~S. Murawski, J.~Ouaknine, B.~Wachter, J.~Worrell, {On the
  Complexity of the Equivalence Problem for Probabilistic Automata}, in:
  FoSSaCS, Vol. 7213 of Lecture Notes in Computer Science, Springer, 2012, pp.
  467--481.

\bibitem{Bellman57}
R.~Bellman, {Dynamic Programming}, Princeton University Press, 1957.

\bibitem{BiancoA95}
A.~Bianco, L.~de~Alfaro, Model checking of probabalistic and nondeterministic
  systems, in: FSTTCS, Vol. 1026 of LNCS, Springer, 1995, pp. 499--513.

\bibitem{KwiatkowskaNP11}
M.~Z. Kwiatkowska, G.~Norman, D.~Parker, {PRISM 4.0: Verification of
  Probabilistic Real-Time Systems}, in: CAV, Vol. 6806 of Lecture Notes in
  Computer Science, Springer, 2011, pp. 585--591.

\bibitem{KatoenZHHJ11}
J.-P. Katoen, I.~S. Zapreev, E.~M. Hahn, H.~Hermanns, D.~N. Jansen, {The ins
  and outs of the probabilistic model checker MRMC}, Perform. Eval. 68~(2)
  (2011) 90--104.

\bibitem{HLSTZ14}
in: FM, Vol. 8442 of Lecture Notes in Computer Science, 2014.

\bibitem{Segala-thesis}
R.~Segala, {Modeling and Verification of Randomized Distributed Realtime
  Systems}, Ph.D. thesis, MIT (1995).

\bibitem{CattaniS02}
S.~Cattani, R.~Segala, {Decision Algorithms for Probabilistic Bisimulation},
  in: CONCUR, Vol. 2421 of Lecture Notes in Computer Science, Springer, 2002,
  pp. 371--385.

\bibitem{BEM00}
C.~Baier, B.~Engelen, M.~E. Majster-Cederbaum, {Deciding Bisimilarity and
  Similarity for Probabilistic Processes}, J. Comput. Syst. Sci. 60~(1) (2000)
  187--231.

\bibitem{HermannsT12}
H.~Hermanns, A.~Turrini, {Deciding Probabilistic Automata Weak Bisimulation in
  Polynomial Time}, in: FSTTCS, Vol.~18 of LIPIcs, Schloss Dagstuhl -
  Leibniz-Zentrum fuer Informatik, 2012, pp. 435--447.

\bibitem{ParmaS07}
A.~Parma, R.~Segala, {Logical Characterizations of Bisimulations for Discrete
  Probabilistic Systems}, in: FoSSaCS, Vol. 4423 of Lecture Notes in Computer
  Science, Springer, 2007, pp. 287--301.

\bibitem{DesharnaisGJP10}
J.~Desharnais, V.~Gupta, R.~Jagadeesan, P.~Panangaden, {Weak bisimulation is
  sound and complete for pCTL$^{*}$}, Inf. Comput. 208~(2) (2010) 203--219.

\bibitem{HermannsPSWZ11}
H.~Hermanns, A.~Parma, R.~Segala, B.~Wachter, L.~Zhang, {Probabilistic Logical
  Characterization}, Inf. Comput. 209~(2) (2011) 154--172.

\bibitem{DoyenHR08}
L.~Doyen, T.~A. Henzinger, J.-F. Raskin, {Equivalence of Labeled Markov
  Chains}, Int. J. Found. Comput. Sci. 19~(3) (2008) 549--563.

\bibitem{EisentrautHZ10}
C.~Eisentraut, H.~Hermanns, L.~Zhang, {On Probabilistic Automata in Continuous
  Time}, in: LICS, IEEE Computer Society, 2010, pp. 342--351.

\bibitem{Hennessy12}
M.~Hennessy, {Exploring probabilistic bisimulations, part I}, Formal Asp.
  Comput. 24~(4-6) (2012) 749--768.

\bibitem{EisentrautGHSZ12}
C.~Eisentraut, J.~C. Godskesen, H.~Hermanns, L.~Song, L.~Zhang, Late weak
  bisimulation for {Markov} automata, CoRR abs/1202.4116.

\bibitem{HermannsKK14}
in: CONCUR 2014, Vol. 8704 of Lecture Notes in Computer Science, 2014.

\bibitem{GiacaloneJS90}
A.~Giacalone, C.~Jou, S.~Smolka, {Algebraic reasoning for probabilistic
  concurrent systems}, in: IFIP TC2 Working Conference on Programming Concepts
  and Methods, North-Holland, 1990, pp. 443--458.

\bibitem{DesharnaisGJP99}
J.~Desharnais, V.~Gupta, R.~Jagadeesan, P.~Panangaden, {Metrics for Labeled
  Markov Systems}, in: CONCUR, Vol. 1664 of Lecture Notes in Computer Science,
  Springer, 1999, pp. 258--273.

\bibitem{AlfaroMRS07}
L.~de~Alfaro, R.~Majumdar, V.~Raman, M.~Stoelinga, Game relations and metrics,
  in: LICS, IEEE Computer Society, 2007, pp. 99--108.

\bibitem{DesharnaisGJP04}
J.~Desharnais, V.~Gupta, R.~Jagadeesan, P.~Panangaden, Metrics for labelled
  {Markov} processes, Theor. Comput. Sci. 318~(3) (2004) 323--354.

\bibitem{DesharnaisLT08}
J.~Desharnais, F.~Laviolette, M.~Tracol, {Approximate Analysis of Probabilistic
  Processes: Logic, Simulation and Games}, in: QEST, IEEE Computer Society,
  2008, pp. 264--273.

\bibitem{FernsPP11}
N.~Ferns, P.~Panangaden, D.~Precup, {Bisimulation Metrics for Continuous Markov
  Decision Processes}, SIAM J. Comput. 40~(6) (2011) 1662--1714.

\bibitem{ChenBW12}
D.~Chen, F.~van Breugel, J.~Worrell, {On the Complexity of Computing
  Probabilistic Bisimilarity}, in: FoSSaCS, Vol. 7213 of Lecture Notes in
  Computer Science, Springer, 2012, pp. 437--451.

\bibitem{Fu12}
H.~Fu, {Computing Game Metrics on Markov Decision Processes}, in: ICALP (2),
  Vol. 7392 of Lecture Notes in Computer Science, Springer, 2012, pp. 227--238.

\bibitem{BacciBLM13}
G.~Bacci, K.~G. Larsen, R.~Mardare, {On-the-Fly Exact Computation of
  Bisimilarity Distances}, in: TACAS, Vol. 7795 of Lecture Notes in Computer
  Science, Springer, 2013, pp. 1--15.

\bibitem{ComaniciPP12}
G.~Comanici, P.~Panangaden, D.~Precup, {On-the-Fly Algorithms for Bisimulation
  Metrics}, in: QEST, IEEE Computer Society, 2012, pp. 94--103.

\bibitem{ChatzikokolakisGPX14}
K.~Chatzikokolakis, D.~Gebler, C.~Palamidessi, L.~Xu, Generalized bisimulation
  metrics, in: {CONCUR}, Vol. 8704 of Lecture Notes in Computer Science, 2014,
  pp. 32--46.

\bibitem{FengZ13}
in: FM, Vol. 8442 of Lecture Notes in Computer Science, 2014.

\bibitem{SegalaL95}
R.~Segala, N.~A. Lynch, {Probabilistic Simulations for Probabilistic
  Processes}, Nord. J. Comput. 2~(2) (1995) 250--273.

\bibitem{SongZG11}
L.~Song, L.~Zhang, J.~C. Godskesen, Bisimulations meet {PCTL} equivalences for
  probabilistic automata, in: {CONCUR}, 2011, pp. 108--123.

\bibitem{LarsenS91}
K.~Larsen, A.~Skou, {Bisimulation through Probabilistic Testing}, Information
  and Computation 94~(1) (1991) 1--28.

\bibitem{Jonsson}
B.~Jonsson, K.~Larsen, Y.~Wang, {Probabilistic extensions of process algebras},
  in: Handbook of Process Algebra, Elsevier, 2001, pp. 685--710.

\bibitem{DArgenioWTC09}
P.~R. D'Argenio, N.~Wolovick, P.~S. Terraf, P.~Celayes, {Nondeterministic
  Labeled Markov Processes: Bisimulations and Logical Characterization}, in:
  QEST, 2009, pp. 11--20.

\bibitem{BernardoNL13}
M.~Bernardo, R.~{De Nicola}, M.~Loreti, A uniform framework for modeling
  nondeterministic, probabilistic, stochastic, or mixed processes and their
  behavioral equivalences, Inf. Comput. 225 (2013) 29--82.

\bibitem{Sto99}
M.~Stoelinga, F.~W. Vaandrager, Root contention in {IEEE} 1394, in: Proceedings
  of the 5th Int.\ AMAST Workshop on Formal Methods for Real-Time and
  Probabilistic Systems, Vol. 1601 of LNCS, Springer, 1999, pp. 53--74.

\bibitem{ZHEJ07}
L.~Zhang, H.~Hermanns, F.~Eisenbrand, D.~N. Jansen, Flow faster: efficient
  decision algorithms for probabilistic simulations, in: TACAS, Vol. 4424 of
  LNCS, Springer, 2007, pp. 155--169.

\bibitem{CrafaR12}
S.~Crafa, F.~Ranzato, Bisimulation and simulation algorithms on probabilistic
  transition systems by abstract interpretation, Formal Methods in System
  Design 40~(3) (2012) 356--376.

\bibitem{SchusterS12}
J.~Schuster, M.~Siegle, Markov automata: Deciding weak bisimulation by means of
  non-na{\"{\i}}vely vanishing states, Inf. Comput. 237 (2014) 151--173.

\bibitem{EisentrautHKT013}
C.~Eisentraut, H.~Hermanns, J.~Kr{\"{a}}mer, A.~Turrini, L.~Zhang, Deciding
  bisimilarities on distributions, in: QEST, 2013, pp. 72--88.

\bibitem{DengGHM09}
Y.~Deng, R.~J. van Glabbeek, M.~Hennessy, C.~Morgan, Testing finitary
  probabilistic processes, in: CONCUR, Vol. 5710 of Lecture Notes in Computer
  Science, Springer, 2009, pp. 274--288.

\bibitem{YingW00}
M.~Ying, M.~Wirsing, Approximate bisimilarity, in: Algebraic Methodology and
  Software Technology, Springer, 2000, pp. 309--322.

\bibitem{Ying01}
M.~Ying, {Topology in Process Calculus: Approximate Correctness and Infinite
  Evolution of Concurrent Programs}, Springer, 2001.

\bibitem{Ying02}
M.~Ying, Bisimulation indexes and their applications, Theoretical Computer
  Science 275 (2002) 1--68.

\bibitem{BreugelSW07}
F.~van Breugel, B.~Sharma, J.~Worrell, {Approximating a Behavioural
  Pseudometric Without Discount for Probabilistic Systems}, in: FoSSaCS, Vol.
  4423 of Lecture Notes in Computer Science, Springer, 2007, pp. 123--137.

\bibitem{ChatterjeeAMR10}
K.~Chatterjee, L.~de~Alfaro, R.~Majumdar, V.~Raman, Algorithms for game metrics
  (full version), Logical Methods in Computer Science 6~(3).

\bibitem{DBLP:journals/corr/abs-1107-1206}
M.~Tracol, J.~Desharnais, A.~Zhioua, Computing distances between probabilistic
  automata, in: QEST, Vol.~57 of {EPTCS}, 2011, pp. 148--162.

\bibitem{DBLP:conf/formats/GiroD07}
S.~Giro, P.~R. D'Argenio, Quantitative model checking revisited: Neither
  decidable nor approximable, in: FORMATS, Vol. 4763 of Lecture Notes in
  Computer Science, Springer, 2007, pp. 179--194.

\bibitem{de1999verification}
L.~De~Alfaro, The verification of probabilistic systems under memoryless
  partial-information policies is hard, Tech. rep., DTIC Document (1999).

\bibitem{Paz:1971:IPA:1097027}
A.~Paz, Introduction to Probabilistic Automata (Computer Science and Applied
  Mathematics), Academic Press, Inc., Orlando, FL, USA, 1971.

\bibitem{Garey:1990:CIG:574848}
M.~R. Garey, D.~S. Johnson, Computers and Intractability: A Guide to the Theory
  of {NP}-Completeness, W. H. Freeman \& Co., New York, NY, USA, 1990.

\bibitem{DBLP:journals/jcss/LyngsoP02}
R.~B. Lyngs{\o}, C.~N.~S. Pedersen, The consensus string problem and the
  complexity of comparing hidden {M}arkov models, J. Comput. Syst. Sci. 65~(3)
  (2002) 545--569.

\end{thebibliography}

\end{document}